\documentclass[12pt]{article}

\usepackage{amsmath}
\usepackage{graphicx,psfrag,epsf}
\usepackage{enumerate}
\usepackage{natbib}
\usepackage{url} 
\usepackage[total={170mm,230mm},
 left=25mm,
 top=20mm]{geometry}
\usepackage{amssymb}
\usepackage[colorlinks=true, allcolors=blue]{hyperref}
\usepackage[table,xcdraw]{xcolor}
\usepackage{booktabs}
\usepackage{array}
\usepackage{bbm}
\newcolumntype{P}[1]{>{\centering\arraybackslash}p{#1}}
\newcolumntype{M}[1]{>{\centering\arraybackslash}m{#1}}
\newcolumntype{L}[1]{>{\raggedright\let\newline\\\arraybackslash\hspace{0pt}}m{#1}}
\newcolumntype{C}[1]{>{\centering\let\newline\\\arraybackslash\hspace{0pt}}m{#1}}
\newcolumntype{R}[1]{>{\raggedleft\let\newline\\\arraybackslash\hspace{0pt}}m{#1}}

\usepackage{amsthm}
\usepackage{amsfonts}
\usepackage{bbm}
\newtheorem{thm}{Theorem}

\newtheorem{prop}[thm]{Proposition}

\usepackage{xpatch}
\usepackage{multirow}
\xpatchcmd{\proof}{\itshape}{\normalfont\proofnamefont}{}{}
\newcommand{\proofnamefont}{\bfseries}
\usepackage{algorithm}
\usepackage{algpseudocode}
\usepackage{caption}

\newcommand{\blind}{1}

\definecolor{ch}{RGB}{30,136,229}
\definecolor{green}{RGB}{34,139,34}

\newcommand\extrafootertext[1]{%
    \bgroup
    \renewcommand\thefootnote{\fnsymbol{footnote}}%
    \renewcommand\thempfootnote{\fnsymbol{mpfootnote}}%
    \footnotetext[0]{#1}%
    \egroup
}

\newcommand\blfootnote[1]{%
  \begingroup
  \renewcommand\thefootnote{}\footnote{#1}%
  \addtocounter{footnote}{-1}%
  \endgroup
}

\algdef{SE}[SUBALG]{Indent}{EndIndent}{}{\algorithmicend\ }%
\algtext*{Indent}
\algtext*{EndIndent}

\begin{document}

\def\spacingset#1{\renewcommand{\baselinestretch}%
{#1}\small\normalsize} \spacingset{1}


\if1\blind
{
  \title{\bf GPDFlow: Generative Multivariate Threshold Exceedance Modeling via Normalizing Flows }}
  \author{Chenglei Hu$^1$* \blfootnote{Corresponding author. Email address: c.hu.2@research.gla.ac.uk}, Daniela Castro-Camilo$^1$}
  
  \footnotetext[1]{
\baselineskip=10pt School of Mathematics and Statistics, University of Glasgow, UK.}

  \maketitle
 \fi

\if0\blind
{
  \bigskip
  \bigskip
  \bigskip
  \begin{center}
    {\LARGE\bf Title}
\end{center}
  \medskip
} \fi

\bigskip
\begin{abstract}

The multivariate generalized Pareto distribution (mGPD) is a common method for modeling extreme threshold exceedance probabilities in environmental and financial risk management. 
Despite its broad applicability, mGPD faces challenges due to the infinite possible parametrizations of its dependence function, with only a few {parametric models} available in practice. 
To address this limitation, we introduce GPDFlow, an innovative mGPD model that leverages normalizing flows to flexibly represent the dependence structure. Unlike traditional parametric mGPD approaches, GPDFlow does not impose explicit parametric assumptions on dependence, resulting in greater flexibility and enhanced performance. 
Additionally, GPDFlow allows direct inference of marginal parameters, providing insights into {marginal} tail {behavior}. 
We derive tail dependence coefficients for GPDFlow, including a bivariate formulation, a $d$-dimensional extension, and an alternative measure for partial exceedance dependence.
A general relationship between the bivariate tail dependence coefficient and the generative samples from normalizing flows is discussed.
Through simulations and a practical application analyzing the risk among five major US banks, we demonstrate that GPDFlow significantly improves modeling accuracy and flexibility compared to traditional parametric methods.
\end{abstract}

\noindent%
\vfill

\newpage
\spacingset{1.45} 
\section{Introduction}
\label{sec:intro}
In multivariate extremes, accurately estimating joint exceedance probabilities is critical for risk management.
Traditionally, this refers to the probability that all components exceed their respective thresholds simultaneously.
However, recent studies increasingly focused on partial joint exceedance probabilities, where only a subset of the components surpasses thresholds \citep{heffernan2004conditional, winter2016modelling, li2024wee}.

Estimating these probabilities falls within the realm of multivariate threshold exceedance modeling, which is typically approached using the multivariate generalized Pareto distribution (mGPD). 
The mGPD naturally extends the univariate generalized Pareto distribution (GPD), a key concept in extreme value theory, as established by \cite{balkema1974residual}. 
Their seminal work led to the Pickands-Balkema-de Haan theorem, which demonstrates that, under certain conditions, the distribution of exceedances above a high threshold converges to a GPD.
The class of distributions satisfying this theorem is linked to the class of distributions described by the Fisher–Tippett–Gnedenko theorem, thereby establishing a connection between the univariate GPD and the univariate generalized extreme value distribution (GEVD). 
This relationship was later extended to the multivariate case by \cite{rootzen2006multivariate}, who also demonstrated that the mGPD is the only threshold-stable multivariate distribution.
While the univariate GPD has been widely applied in threshold exceedance models \citep{davison_models_1990, chavez2005generalized, castro2019spliced, he2022risk}, its multivariate counterpart has seen relatively less use, primarily due to the challenges in extending the univariate formulation to the multivariate case. 
This is mainly due to the absence of a unique parametrization, as the dependence structure in the multivariate GEVD allows for infinitely many possible parameterizations \citep{coles2001introduction}. 
To overcome this, \cite{rootzen2006multivariate, rootzen2018multivariate, rootzen2018multivariate_b} and \cite{kiriliouk2019peaks} proposed several parametrizations specifically designed to yield closed-form expressions that facilitate efficient computation.
However, these parametrizations often fail to fully capture the complexity of real-world data.


To address this, we propose GPDFlow, an innovative flow-based mGPD model that employs normalizing flows to flexibly represent the dependence structure. 
Normalizing flows are powerful deep generative models that transform simple base distributions (such as Gaussian distributions) through a sequence of invertible mappings, enabling the approximation of highly complex distributions \citep{papamakarios2021normalizing}. 
By leveraging normalizing flows, GPDFlow avoids restrictive assumptions on tail dependencies, allowing the data to dictate dependence structures organically. 
Furthermore, embedding normalizing flows within the mGPD framework makes GPDFlow statistically rigorous and maintains the desirable theoretical properties of traditional mGPD models.

The GPDFlow model offers two key advantages.
First, it simultaneously captures marginal distributions and dependence structures within a unified modeling framework, significantly improving existing two-step extremal dependence modeling approaches that use generative approaches \citep{boulaguiem2022modeling, mcdonald2022comet}. 
The marginal scale and shape parameters are explicitly estimated outside the flow transformations, providing direct insight into tail heaviness and overcoming known difficulties of affine triangular-based flow models in modeling heavy tails \citep{jaini2020tails}.
Second, GPDFlow serves as a proper statistical distribution with an explicit density function, facilitating straightforward maximum likelihood estimation and efficient sampling. 
Consequently, statistical inference—including extrapolation via Monte Carlo methods—is both reliable and computationally practical.

The rest of the paper is organized as follows.
Section \ref{sec: method} provides the theoretical foundations of mGPD, normalizing flows, and our proposed GPDFlow model. 
In Section \ref{sec: Inference}, we present the inference procedure and derive an explicit formulation for the bivariate tail dependence coefficient of a GPDFlow. We further extend this analysis to the $d$-dimensional case ($d > 2$) and introduce an alternative coefficient designed to quantify the dependence of partial exceedances. Besides, we explore how the GPDFlow tail dependence coefficient relates to the broader properties of samples generated by normalizing flows embedded within GPDFlow for $d=2$. 
Simulation studies in Section \ref{sec: simulation} assess GPDFlow’s performance under both correctly specified and misspecified conditions. 
Section \ref{sec: application} illustrates the practical utility of GPDFlow in financial risk management, specifically through calculating the Conditional Value at Risk (CoVaR) for negative returns of five major U.S. banks. 
Lastly, Section \ref{sec: Discussion} acknowledges existing limitations of GPDFlow and outlines practical strategies to address these challenges.

\textbf{Notation}: Throughout the paper, $\max(\cdot)$ and $\min(\cdot)$ take one vector as input, returning its largest or smallest element, respectively, while 
$\wedge$ operates on two vectors, returning their elementwise minima.
All operations between two vectors $\boldsymbol{x}=(x_1,\cdots,x_d)$ and $\boldsymbol{y}=(y_1,\cdots,y_d)$ are performed elementwise. 
For example, $\boldsymbol{x}\boldsymbol{y}$ denotes the Hadamard product of $\boldsymbol{x}$ and $\boldsymbol{y}$, i.e. $\boldsymbol{xy}=(x_1y_1,\cdots,x_dy_d)$, and $\boldsymbol{x} \wedge \boldsymbol{y}$ represents the elementwise minima of $\boldsymbol{x}$ and $\boldsymbol{y}$, where $(\boldsymbol{x}\wedge \boldsymbol{y})_{j} = \min\{(x_j,y_j)\},\; j=1,\cdots,d $.
Similarly, logarithm, exponential, and indicator functions are applied elementwise when operating on vectors.


\section{Method}
\label{sec: method}

\subsection{Multivariate GPD}
Suppose $\boldsymbol{Y}$ is a $d$-dimensional vector with cumulative distribution function $F$ that is in the max-domain of attraction of a non-degenerated distribution $G$.
This is to say, there exists sequences $\boldsymbol{a}_n \in (0,\infty)^d $ and $\boldsymbol{b}_n \in \mathbb{R}^d$ such that 
\begin{align}
    \lim_{n \rightarrow \infty } n\{1-F(\boldsymbol{a}_n\boldsymbol{y} + \boldsymbol{b}_n)\} = -\log G(\boldsymbol{y}).
    \label{eq: MDA}
\end{align}
The multivariate generalized Pareto distribution (mGPD) $H(\boldsymbol{x})$ arises when considering the conditional probability of $\boldsymbol{Y}$ when at least one component of $\boldsymbol{Y}$ is extreme \citep{rootzen2006multivariate}. 
Formally,
\begin{equation}
    \begin{split}
     H(\boldsymbol{x})&= \lim_{n\to \infty} \mathrm{P}\{\boldsymbol{a}_n^{-1}(\boldsymbol{Y}-\boldsymbol{b}_n) \leq \boldsymbol{x} | \boldsymbol{Y} \nleqslant \boldsymbol{b}_n\}\\
&=\frac{1}{\log G(\boldsymbol{0})} \log\frac{G(\boldsymbol{x}\wedge\boldsymbol{0} )}{G(\boldsymbol{x})}   
    \end{split}    
    \label{eq: GPD}
\end{equation}
The distribution $G$ in \eqref{eq: MDA} and \eqref{eq: GPD} is the multivariate generalized extreme value distribution. 
As a direct result of the weak convergence in \eqref{eq: MDA}, which implies the convergence of both margin and copula, the marginal distributions $G_j(x_j)$, $j=1,\ldots,d,$ are univariate generalized extreme value distributions by Fisher–Tippett–Gnedenko theorem \citep[p.~6]{haan2006extreme}, with distribution function
\begin{align*}
    G_j\left(x_j\right)= \begin{cases}\exp \left\{-\left\{1+\gamma_j\left(x_j-\mu_j\right) / \alpha_j\right\}^{-1 / \gamma_j}\right\} & \text { if } \gamma_j \neq 0, \\ \exp \left\{-\exp \left\{-\left(x_j-\mu_j\right) / \alpha_j\right\}\right\} & \text { if } \gamma_j=0,\end{cases}
\end{align*}
where $\alpha_j \in (0,\infty), \mu_j,\gamma_j \in \mathbb{R}$ and support $\{x_j \in \mathbb{R}: \alpha_j + \gamma_j(x_j - \mu_j)>0\} $.
Appropriate choices of $\boldsymbol{a}_n$ and $\boldsymbol{b}_n$ always ensure 0 is in the support of $G_j$.
The convergence of the corresponding copula implies 
\begin{align*}
    \lim_{n \rightarrow \infty } n\left\{1 - C_F\left(1-\frac{x_1}{n},\cdots, 1-\frac{x_d}{n}\right)\right\} = -\log C_G \left(\exp\{-x_1\},\cdots, \exp\{-x_d\}\right) := \ell (\boldsymbol{x})
\end{align*}
where $C_F$ and $C_G$ are the copulas of $F$ and $G$, respectively, and $\ell(\boldsymbol{x})$ is called the stable tail dependence function (stdf) of $C_G$, which describes the dependence structure of $G$ \citep{segers2012max}.
The stdf $\ell(\boldsymbol{x})$ does not have finite representations.
In fact,  any function $\ell(\boldsymbol{x}): [0,\infty)^d \rightarrow [0,\infty)$ is a valid stdf if it can be written as 
\begin{align}
    \ell(\boldsymbol{x}) = \mathbb{E}[\max(\boldsymbol{x}\boldsymbol{V})],
    \label{eq: D_norm}
\end{align}
where $\boldsymbol{V}$ is a $d$-dimensional random variable on $[0,\infty)^d$ satisfying $\mathbb{E}(V_j)=1, \; j=1,\cdots,d$ \citep{rootzen2018multivariate}. 
The expression in \eqref{eq: D_norm} defines $\ell(\boldsymbol{x})$ as a norm, call D-Norm, and $\boldsymbol{V}$ is the generator of $\ell(\boldsymbol{x})$ \citep{FALK201785}.

Combining the margins and the stdf, we can express $G(\boldsymbol{x})$ as 
\begin{align*}
    G(\boldsymbol{x}) = \exp\{-\ell\{-\log G_1(x_1),\cdots,-\log G_d(x_d)\} \}.
\end{align*}
Consequently, $H(\boldsymbol{x})$ is determined by the
marginal parameters $\boldsymbol{\alpha} = (\alpha_1,\cdots,\alpha_d), \boldsymbol{\mu} = (\mu_1,\cdots,\mu_d)$, $\boldsymbol{\gamma}=(\gamma_1,\cdots,\gamma_d)$ and the stdf $\ell(\boldsymbol{x})$. 
As noted by \cite{rootzen2018multivariate}, the parametrization of $H$ in terms of $(\boldsymbol{\alpha}, \boldsymbol{\mu},\boldsymbol{\gamma},\ell)$ is not convenient since $\boldsymbol{\alpha}$ and $\boldsymbol{\mu}$ are not identifiable from $H$ due to the max-stable property of $G$. 
This could be addressed by reparameterizing $H$ by $(\boldsymbol{\sigma},\boldsymbol{\gamma},\ell)$, where $\boldsymbol{\sigma} = \boldsymbol{\alpha} - \boldsymbol{\gamma}\boldsymbol{\mu}$
can be seen as the marginal scale parameter of $H(\boldsymbol{x})$.
Note that ${\sigma}_j > {0}$ is ensured by $G_j(0)>0$. 

If a random vector $\boldsymbol{X}$ follows a mGPD, then the marginal distributions 
\begin{align}
    H_j(x_j) = \frac{1}{\log G(\boldsymbol{0})} \log\frac{G((0,\cdots,0,x_j \wedge0,0,\cdots,0))}{G_j(x_j)}, \quad \sigma_j+\gamma_jx_j >0, \;j=1,\cdots,d,
    \label{eq: marginal_GPD}
\end{align}
depend on the marginal parameters $\sigma_j$, $\gamma_j$ and stdf $\ell$.
But the conditional distribution of $X_j>x_j|X_j>0$ relies only on  $\sigma_j$ and $\gamma_j$:
\begin{align}
  \frac{1- H_j(x_j)}{1-H_j(0)} =\frac{\log G_j(x_j)}{\log G_j(0)}= \left(1+\frac{
  \gamma_jx_j
  }{\sigma_j}\right)^{-1/\gamma_j}, \quad x_j>0.
  \label{eq: marginal_conditional_mgpd}
\end{align}
In other words, for an mGPD random vector, $\boldsymbol{X}$, the conditional distribution of $X_j>x_j|X_j>0$ is a univariate GPD.
Consequently, the parameters $\sigma_j$ and $\gamma_j$ retain the same interpretations as the scale and shape parameters of a univariate GPD.
What's more, we can use the following transformation 
\begin{align}
    \boldsymbol{Z} =g_{\text{std}}(\boldsymbol{X};\boldsymbol{\sigma},\boldsymbol{\gamma})= \mathbbm{1}\{\boldsymbol{\gamma} \neq \boldsymbol{0}\} \frac{1}{\boldsymbol{\gamma}}\log \left(\frac{\boldsymbol{\gamma}\boldsymbol{X}}{\boldsymbol{\sigma}}+1 \right) + \mathbbm{1}\{\boldsymbol{\gamma} = \boldsymbol{0}\}\frac{\boldsymbol{X}}{\boldsymbol{\sigma}}
    \label{eq: transformation}
\end{align}
to standardize the margins so that $\boldsymbol{\sigma} =\boldsymbol{1} $,  $\boldsymbol{\gamma}= \boldsymbol{0}$, and $\mathbb{P}(Z_j>z_j|Z_j>0) = \exp(-z_j)$ for $z_j>0$. We call this form the \emph{standardized mGPD} and denote its cdf and density as $H(\boldsymbol{z})$ and $h(\boldsymbol{z})$, respectively.

Although $H(\boldsymbol{z})$ can be expressed as a function of $\boldsymbol{V}$ through \eqref{eq: GPD}, deriving $h(\boldsymbol{z})$ is less straightforward. 
Alternatively, $H(\boldsymbol{z})$ and $h(\boldsymbol{z})$ can be derived using a stochastic representation of the mGPD called the T representation \citep{rootzen2018multivariate}, given by
\begin{align}
    \boldsymbol{Z} = E + \boldsymbol{T} - \max({\boldsymbol{T})},
    \label{eq: stocastic_representation}
\end{align}
where $E$ is a unit exponential random variable, and $\boldsymbol{T}$ is a $d$-dimensional random variable independent to $E$ and satisfying the weak conditions $\mathrm{P}(T_j > -\infty) > 0$ and $\mathrm{P}(\max(\boldsymbol{T}) > -\infty) = 1$. 
$\boldsymbol{T}$ can be seen as a generator for the mGPD, which jointly influences the dependence and margins of $\boldsymbol{Z}$ with $E$.
The generator $\boldsymbol{V}$ in $\ell(\boldsymbol{x})$ is associated with $\boldsymbol{T}$ by $\boldsymbol{V} = \exp\{\boldsymbol{T}-\max(\boldsymbol{T})\}/\mathbb{E}(\exp\{\boldsymbol{T}-\max(\boldsymbol{T})\})$.
Using \eqref{eq: stocastic_representation}, the density function $h(\boldsymbol{z})$ is given by
\begin{align}
    h(\boldsymbol{z}) = \frac{\mathbbm{1}\{\max(\boldsymbol{z}>0)\}}{\exp\{\max{(\boldsymbol{z})}\}}\int_{-\infty}^{\infty}f_{\boldsymbol{T}}(\boldsymbol{z}+s) \mathrm{d} s,
    \label{eq: mGPD_Density}
\end{align}
where $f_{\boldsymbol{T}}$ is the density of $\boldsymbol{T}$.
Since both conditions of $\boldsymbol{T}$ are satisfied by most common distributions, its flexibility allows for various parametrizations of the mGPD.
For instance, \cite{kiriliouk2019peaks} studied the cases when $\boldsymbol{T}$ follows a multivariate Gaussian and several multivariate distributions with independent components (Gumbel, reverse Gumbel, reverse exponential, and log-gamma), in which closed-form expressions of $h(\boldsymbol{z})$ can be obtained.

There exist alternative representations of the mGPD based on point process representations of extreme episodes~\citep{rootzen2018multivariate,rootzen2018multivariate_b}.
For instance, the R representation is derived directly from the point process in Section 3 of~\cite{rootzen2018multivariate_b}, while the U representation transforms R into a standardized scale by \eqref{eq: transformation}.
Although both the R representation and U representation give rise to densities $h(\boldsymbol{z})$ slightly different from \eqref{eq: mGPD_Density}, they can be transformed into the T representation by adding constraints to the corresponding generators.
Additionally, Theorem 4.4 in \cite{rootzen2018multivariate} highlighted that for any standardized mGPD, there always exists a random vector $\boldsymbol{T}$ capable of generating it.
Therefore, without losing generality, we use the T representation here due to its simple density form and straightforward sampling.

\subsection{Normalizing Flows}
Normalizing flows are flow-based generative models that apply a sequence of changes of variables to convert a simple distribution (base distribution) to any well-behaved distribution (target distribution).
Starting with a $d$-dimensional random variable $\boldsymbol{U}$ with base density $f_{\boldsymbol{U}}(\boldsymbol{u})$ and a invertible and differentiable transformation $g$ (namely a diffeomorphism), the transformed random variable $\boldsymbol{Y} = g(\boldsymbol{U}) $ has density 
\begin{align}
    f_{\boldsymbol{Y}}(\boldsymbol{y}) = f_{\boldsymbol{U}}(\boldsymbol{u}) |\det J_g(\boldsymbol{u})|^{-1}, \quad \boldsymbol{u} = g^{-1}(\boldsymbol{y}),
    \label{eq: change_of_variable}
\end{align}
where $J_g(\boldsymbol{u})$ is the $d \times d$ Jacobian matrix of $g$, i.e.,
\begin{align*}
J_g(\mathbf{u}) =
\begin{bmatrix}
\frac{\partial g_1}{\partial u_1} & \cdots & \frac{\partial g_1}{\partial u_d} \\
\vdots & \ddots & \vdots \\
\frac{\partial g_d}{\partial u_1} & \cdots & \frac{\partial g_d}{\partial u_d}
\end{bmatrix}.
\end{align*}
In the generative model context, \eqref{eq: change_of_variable} represents the generative direction, where ``noise'' from the base distribution is transformed into the target distribution. 
This process also corresponds to the sampling mechanism in a flow-based model.
The reversed process of \eqref{eq: change_of_variable} is the normalizing process, which transforms a target sample into the base sample and provides information for updating $J_g$. 
Such diffeomorphism $g$ always exists for any well-behavior $f_{\boldsymbol{Y}}(\boldsymbol{y})$ and $f_{\boldsymbol{U}}(\boldsymbol{u})$  \citep{papamakarios2021normalizing}. 
To construct a $g$ that is expressive enough in the sense that it can be used to transform $f_{\boldsymbol{U}}$ to any $f_{\boldsymbol{Y}}$, we can first use a neural network to represent a diffeomorphism $g_i$, $i=1,\ldots,K$, where its expressive power is justified by the universal approximation theorem \citep{hornik1989multilayer}.
We can then compose all $g_i$'s into a single transformation $g = g_K\circ  \cdots  \circ g_1$ to further improve the expressivity of a single $g_i$.

A few architectures for $g_i$ have been studied to ensure its differentiability and invertibility \citep{dinh2017density,kingma2018glow, papamakarios2017masked}.
In this paper, we use the real-valued non-volume preserving (Real NVP; \citealt{dinh2017density}), which constructs complex probability distributions by applying a sequence of invertible, affine transformations with a triangular Jacobian structure to a base distribution, enabling efficient density evaluation and sampling.
Specifically, the Real NVP model constructs a transformation $g$ by stacking multiple coupling layers together. In each coupling layer, an input vector $\boldsymbol{u}$ is partitioned into two parts $(\boldsymbol{u}_{p_1},\boldsymbol{u}_{p_2})$.
Then $\boldsymbol{u}_{p_1}$ remains unchanged and an affine transformation dependent on $\boldsymbol{u}_{p_1}$ is applied to $\boldsymbol{u}_{p_2}$. The mathematical formulation for such single coupling layer transformation $g_k: \boldsymbol{y} = g_k(\boldsymbol{u})$ is given by
\begin{equation}
\begin{aligned}
\boldsymbol{y} = \boldsymbol{b} \boldsymbol{u} + (\boldsymbol{1} - \boldsymbol{b})  \left\{ \boldsymbol{u}  \exp\left\{\zeta_k(\boldsymbol{b}  \boldsymbol{u})\right\} + \upsilon_k (\boldsymbol{b}  \boldsymbol{u}) \right\},
\end{aligned}
\label{eq: Real_NVP}
\end{equation}
where $\zeta_k,\upsilon_k: \mathbb{R}^d \rightarrow \mathbb{R}^d$ {are, respectively, the log-scale and translation functions, defined as} multilayer perceptions (MLPs). 
MLP is a feedforward neural network that extends regression by learning nonlinear transformations via weighted sums and activation functions, optimized through backpropagation~\citep{rumelhart1986learning}.
The vector $\boldsymbol{b}$ is a $d$-dimensional binary mask vector that masks components, i.e. determine which components are $\boldsymbol{u}_{p_1}$.
Let $M= \{j =1,\cdots,d:b_j=1\}$ denote the set of masked indices and $N = \{1,\cdots,d\}\backslash M$.
Then, \eqref{eq: Real_NVP} implies
\begin{equation}
\begin{aligned}
\boldsymbol{y}_{j \in M} &= \boldsymbol{u}_{j \in M}, \\
\boldsymbol{y}_{j \in N} &= \boldsymbol{u}_{j \in N}  \exp\{\zeta_k(\boldsymbol{b}  \boldsymbol{u})\}_{j \in N} + \upsilon_k(\boldsymbol{b}  \boldsymbol{u})_{j \in N}.
\end{aligned}
\label{eq: Real_NVP_decompose}
\end{equation}
The components in the masked set $M$ are invariant under the transformation, while the components in the unmasked set $N$ are applied to an affine transformation with scale $\exp\{\zeta_k(\boldsymbol{b}  \boldsymbol{u})\}_{j \in N}$ and translation  $\upsilon_k(\boldsymbol{b}  \boldsymbol{u})_{j \in N}$. 
Due to the masking, $ \boldsymbol{b}  \boldsymbol{u}$ is a $d$-dimensional vector that has $u_j,\; j\in M$ at the $j$-th component and 0 elsewhere. 
Hence, the scale $\exp\{\zeta_k(\boldsymbol{b}  \boldsymbol{u})\}_{j \in N}$ and translation   $\upsilon_k(\boldsymbol{b}  \boldsymbol{u})_{j \in N}$ are functions of $\boldsymbol{u}_{j \in M}$, and by updating these functions, Real NVP will learn the dependence information between $\boldsymbol{u}_{j\in M}$ and $\boldsymbol{u}_{j\in N}$. 
To ensure comprehensive learning of the dependence, $\boldsymbol{b}$ is usually flipped in every coupling layer (i.e., we swap $M$ and $N$).

An appealing feature of \eqref{eq: Real_NVP_decompose} is that the determinant of the Jacobian of $g_k$'s can be easily calculated by rearranging the Jacobian into a lower triangle matrix:
\begin{equation*}
    \left| \det(g_k)  \right| = \left| \det\left (\begin{bmatrix}
\mathbb{I}_{|M|} & \boldsymbol{0} \\
\boldsymbol{L}_{|N|\times|M|}(\boldsymbol{u}_{j \in M})& \text{diag}(\exp\{\zeta_k(\boldsymbol{b}\boldsymbol{u})\}_{j\in N})
\end{bmatrix} \right)  \right|=\prod_{j \in N}\exp\left\{\zeta_k(\boldsymbol{bu})\right\}_j.
\end{equation*}
Since  $|\det(g_k)|$ only depends on the output of $\zeta_k$ and is irrelevant to the specific structure of $\zeta_k$ and $\upsilon_k$, we can define $\zeta_k$ and $\upsilon_k$ as complex as needed without worrying about the differentiability or invertibility of $g_i$.

Suppose we specify a $K$-layer transformation $g = g_K\circ  \cdots  \circ g_1$ with $g_k$ defined {as in \eqref{eq: Real_NVP}}. 
Let $\boldsymbol{\theta}$ denote all weights in $\zeta_k$ and $\upsilon_k$, $k=1,\cdots,K$. 
Then, the estimation of a Real NVP model boils down to finding $\widehat{\boldsymbol{\theta}}$ that maximize the likelihood $f_{\boldsymbol{Y}}(\boldsymbol{y};\boldsymbol{\theta})$ given a known  $f_{\boldsymbol{U}}(\boldsymbol{u})$.
Gradient descent and backpropagation are usually used to update $\boldsymbol{\theta}$ and find the maximum likelihood.
Due to the flexibility of $g$,  $f_{\boldsymbol{U}}(\boldsymbol{u})$ does not have a huge influence on the expressivity of $f_{\boldsymbol{Y}}(\boldsymbol{y})$ in most situations, so  $f_{\boldsymbol{U}}(\boldsymbol{u})$ is commonly set as a standard multivariate Gaussian \citep{dinh2017density}.
The above reveals that the RealNVP model can provide an exact density of $f_{\boldsymbol{Y}}(\boldsymbol{y};\boldsymbol{\theta})$, and estimating a RealNVP model is much like fitting a statistical distribution.
This exact likelihood inference is one of the main distinctions of normalizing flows from other generative models.



\subsection{GPDFlow}
We propose a flow-based mGPD called GPDFlow by expressing $f_{\boldsymbol{T}}$ in \eqref{eq: mGPD_Density} using Real NVP.
Specifically, let $\boldsymbol{x}$ denote a $d$-dimensional threshold exceedance observation, and $\boldsymbol{z}$ be a standardized version of $\boldsymbol{x}$, i.e. $\boldsymbol{z} = g_{\text{std}}(\boldsymbol{x};\boldsymbol{\sigma},\boldsymbol{\gamma})$ as defined in~\eqref{eq: transformation}, GPDFlow is a $d$-dimensional distribution with density function
\begin{align}
    f(\boldsymbol{x};\boldsymbol{\sigma},\boldsymbol{\gamma},\boldsymbol{\theta}) = \frac{\mathbbm{1}\{\max(\boldsymbol{z}>0)\}}{\exp\{\max{(\boldsymbol{z})}\}}\int_{-\infty}^{\infty}f_{\boldsymbol{U}}\{g^{-1}(\boldsymbol{z}+s;\boldsymbol{\theta})\}|\det J_g(\boldsymbol{z}+s;\boldsymbol{\theta})|^{-1} \mathrm{d} s \prod_{j=1}^d\frac{1}{\sigma_j+\gamma_jx_j},
    \label{eq: GPDFlow_density}
\end{align}
where $f_{\boldsymbol{U}}(\boldsymbol{u})$ is the density of a $d$-dimensional standard Gaussian distribution, $\boldsymbol{\sigma}$, $\boldsymbol{\gamma}$ are marginal parameters, $g$ is the entire transformation in normalizing flows with $\boldsymbol{\theta}$ as the total weights.

Equation \eqref{eq: GPDFlow_density} can be understood in two ways. 
From a statistical perspective, we are parametrizing $f_{\boldsymbol{T}}$ in \eqref{eq: mGPD_Density} by normalizing flows; hence, GPDFlow is a valid mGPD and inherits all useful properties from it, e.g. threshold stability, lower dimensional conditional margins, and sum-stability under shape constraints \citep{kiriliouk2019peaks}, but allows a more flexible dependence modeling. 
Viewing  \eqref{eq: GPDFlow_density}  from a normalizing flows angle, a random variable $\boldsymbol{X}$ with density $f$ can be obtained by 
\begin{align}
    \boldsymbol{X} = g_{\text{std}}\circ g_{\text{mGPD}} \circ g(\boldsymbol{U}),
\end{align}
where $ g_{\text{mGPD}}(\boldsymbol{T}) = E+ \boldsymbol{T} - \max{(\boldsymbol{T})}, \; E \sim \exp(1)$.
For the output $g(\boldsymbol{U})$ from Real NVP, the transformation $g_{\text{mGPD}}$ \emph{squeezes} it into a reversed L-shape (see Figure \ref{fig: simulation_2}), ensuring that the output satisfies the mGPD's properties.
The final layer transformation $g_{\text{std}}$ {adjusts for the scale of the observations via $\boldsymbol{\sigma}$ and accounts for marginal tail heaviness through $\boldsymbol{\gamma}$.
A key advantage of the above structure is that it prevents Real NVP from directly modeling the marginal tail, where standard Real NVP often struggles to accurately estimate tail heaviness \citep{jaini2020tails}.}


Unlike recent two-step generative approaches in the extreme value modeling, which first transform margins to a standard scale before estimating the dependence~\citep{boulaguiem2022modeling}, GPDFlow jointly estimates both the marginal distribution and tail dependence for threshold exceedances.
The marginal distribution of the mGPD in \eqref{eq: marginal_GPD} depends on $\ell$ and, consequently, on $\boldsymbol{T}$. 
As a result, the flexibility of $f_{\boldsymbol{T}}$ in GPDFlow helps reduce the bias in marginal estimations compared to the classic mGPD, which often struggles to balance marginal and tail dependence estimation. 
Traditional mGPD models typically rely on the censored likelihood method to alleviate the bias introduced by an inappropriate description of the lower tail margins \citep{kiriliouk2019peaks}.  However, this approach can become computationally expensive in high dimensions.

\section{Inference}
\label{sec: Inference}
\subsection{Likelihood and model estimation }
Given $d$-dimensional observations $\boldsymbol{x}_i = (x_{i1},\cdots, x_{id})$, $i=1,\cdots,n$, we estimate the parameter $(\boldsymbol{\sigma},\boldsymbol{\gamma},\boldsymbol{\theta})$ in GPDFlow by maximizing the full log-likelihood 
\begin{equation}
  \begin{aligned}
l(\boldsymbol{\sigma},\boldsymbol{\gamma},\boldsymbol{\theta}) &=\sum_i^n \left\{ \log\left\{\int_{-\infty}^{\infty}f_{\boldsymbol{U}}\{g^{-1}(\boldsymbol{z}_i+s;\boldsymbol{\theta})\}|\det J_g(\boldsymbol{z}_i+s;\boldsymbol{\theta})|^{-1} \mathrm{d} s \right\} \right.\\
&\left.+ \log{\mathbbm{1}\{\max(\boldsymbol{z}_i>0)\}}- \max{(\boldsymbol{z}_i)} - \sum_{j=1}^d\log(|\sigma_j + \gamma_jx_{ij}|) \right\},
  \quad \boldsymbol{z}_i = g_{\text{std}}(\boldsymbol{x}_i).
\end{aligned}  
\label{eq: loglikilihood_GPDFlow}
\end{equation}
The integrant in \eqref{eq: loglikilihood_GPDFlow}, which is the flow density, maps $\mathbb{R}^d$ to $\mathbb{R}$. 
However, the integral is univariate since the scalar $s$ is added to each component of $\boldsymbol{z}_i$ and the integration is over $s$.
As a result, the computation complexity does not scale with the dimension, and the integral can be effectively approximated using numerical methods such as the trapezoidal rule or Monte Carlo. 
In the practical implementation, $\boldsymbol{z}_i$ is expressed as 
$$\boldsymbol{z}_i = \mathbbm{1}\{\boldsymbol{\gamma} \neq \boldsymbol{0}\} \boldsymbol{\gamma}^{-1}\log \left(|\boldsymbol{\gamma}\boldsymbol{x}_i/\boldsymbol{\sigma}+1| \right) + \mathbbm{1}\{\boldsymbol{\gamma} = \boldsymbol{0}\}\boldsymbol{x}_i/\boldsymbol{\sigma}$$ to avoid taking the logarithm of a negative value.
Additionally, an extra penalty term $$ \lambda \sum_i \sum_j\mathbbm{1}\{\sigma_j+\gamma_jx_{ij}\leq0\}(\sigma_j+\gamma_j x_{ij})^2$$ is added to \eqref{eq: loglikilihood_GPDFlow} to penalize values of $\sigma_j$ and $\gamma_j$ that do not satisfy $\sigma_j + \gamma_jx_j >0$. 
Here, $\lambda$ is a hyperparameter that controls the penalty and is typically set to a large value.

An interesting characteristic of $\eqref{eq: mGPD_Density}$ is that $f_{\boldsymbol{T}}$ cannot be uniquely identified from $h(\boldsymbol{z})$.
This could be straightforwardly seen from $\eqref{eq: stocastic_representation}$ by noticing that, for any $d$-dimensional random vector $\boldsymbol{R}=(R,\cdots,R)$,
$\boldsymbol{T}$ and $\boldsymbol{T} + \boldsymbol{R}$ will always lead to same $\boldsymbol{Z}$.
The expression, $\boldsymbol{T} - \max{(\boldsymbol{T}) :=\boldsymbol{S}}$ is called spectral random vector in \cite{rootzen2018multivariate}, and its density $f_{\boldsymbol{S}}$ can be identified from $h(\boldsymbol{z})$.
A direct consequence of the unidentifiability of $f_{\boldsymbol{T}}$ is that $f_{\boldsymbol{T}}$ could have different estimated densities depending on the initial weights of Real NVP, an issue illustrated in Figure \ref{fig: f_T_sample}.
The reason we still build GPDFlow on $f_{\boldsymbol{T}}$ rather than $f_{\boldsymbol{S}}$ is that $f_{\boldsymbol{T}}$ is an unbounded density, whereas the support of $f_{\boldsymbol{S}}$ is defined by the union of the coordinate hyperplanes where at least one coordinate is zero. Applying such constraints of support for $f_{\boldsymbol{S}}$ in normalizing flows is difficult; by contrast, no constraints are required in the $f_{\boldsymbol{T}}$ modeling. Since we are only interested in $h(\boldsymbol{z})$ in the threshold exceedance modeling rather than specific $f_{\boldsymbol{T}} $, the unidentifiability of $f_{\boldsymbol{T}}$ is not an issue in the practical use of GPDFlow. 
\begin{figure}[H]
\begin{center}
\includegraphics[width=5in]{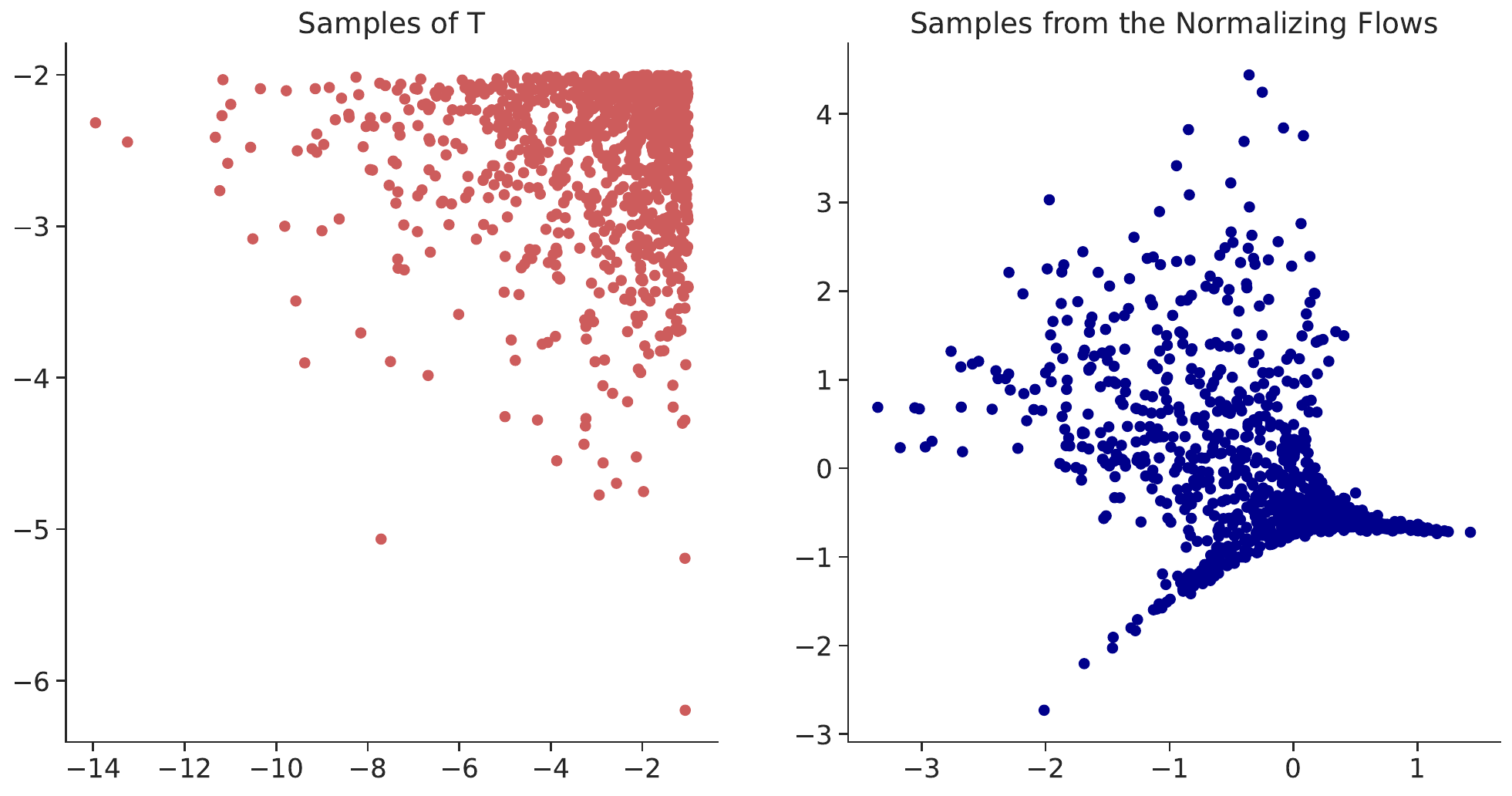}
\end{center}
\caption{\footnotesize{Comparison between 1000 samples (left) from $f_{\boldsymbol{T}}(t_1,t_2) = \exp\{(t_1+1)/2 + (t_2+2)/0.5\}\\ t_1<-1, t_2<-2$ , and samples (right) from estimated $\widehat{f}_{\boldsymbol{T}}$ from the mGPD data generated by $f_{\boldsymbol{T}}$, illustrating the unidentifiability of $f_{\boldsymbol{T}}$. }}
\label{fig: f_T_sample}
\end{figure}
Even though $f_{\boldsymbol{T}}$ is not identifiable, we aim to identify characteristics of the mGPD associated with $f_{\boldsymbol{T}}$, through which we can anticipate the behavior of samples from the fitted normalizing flows.
The following proposition addresses this by examining the bivariate tail dependence coefficient of the mGPD.

\begin{prop}
Suppose $F_1$ and $F_2$ are the marginal cdfs of a bivariate mGPD with generator $\boldsymbol{T}=(T_1,T_2)$. 
Define $q^*$ as the quantile level from which both components exceed the threshold, i.e., $q^* = \inf\{q\in(0,1):F_1^{-1}(q)>0 \text{ and } F_2^{-1}(q)>0 \}$ and let $
\chi_{1,2}(q) = \mathbb{P}(X_1>F_1^{-1}(q) | X_2>F_2^{-1}(q))$.
Then, the bivariate tail dependence coefficient of the mGPD, $\chi_{1,2} = \lim_{q \rightarrow 1^-} \chi_{1,2}(q)$, can be expressed as 
\begin{align}\label{eq:chi2}
    \chi_{1,2} =
\mathbb{E}\left(\min \left\{\frac{\exp\{T_1-\max(\boldsymbol{T})\}}{\mathbb{E}(\exp\{T_1-\max(\boldsymbol{T})\})},\frac{\exp\{T_2-\max(\boldsymbol{T})\}}{\mathbb{E}(\exp\{T_2-\max(\boldsymbol{T})\})}\right\}\right).
\end{align}
Specifically, $$\chi_{1,2}(q\mid q>q^*)=\chi_{1,2}.$$
If $T_1$ and $T_2$ are exchangeable, $\chi_{1,2}$ can be simplified to 
\begin{equation*}
   \chi_{1,2} =\chi_{1,2}(q\mid q>q^*)= \frac{2\mathbb{E}(\exp\{-|T_1-T_2|\})}{1+\mathbb{E}(\exp\{-|T_1-T_2|\})}.
\end{equation*}
\label{prop: chi_mGPD}
\end{prop}
\vspace{-1cm}

Proposition \ref{prop: chi_mGPD} states that $\chi_{1,2}(q\mid q>q^*)$ is a positive constant and is equal to the bivariate tail dependence coefficient. 
When estimating GPDFlow, this constant line pattern of $\chi_{1,2}(q\mid q>q^*)$ is preserved regardless of the normalizing flows output, and the value of $\chi_{1,2}(q\mid q>q^*)$ or $\chi_{1,2}$ determines the general behavior of $f_{\boldsymbol{T}}$.
{Under the exchangeability assumption, $\chi_{1,2}$ depends solely on the distribution of the difference between $T_1$ and $T_2$.
If the data exhibits weak tail dependence (i.e., small $\chi_{1,2}$), $T_1$ and $T_2$ tend to spread across their sample spaces, exhibiting large differences.
In contrast, strong tail dependence (large $\chi_{1,2}$) implies that $T_1$ and $T_2$ are more concentrated around similar values, and display smaller discrepancies. 
In the extreme case where $\chi_{1,2}=1$, we have $T_1=T_2$ almost surely.}
A proof for the proposition \ref{prop: chi_mGPD} can be found in Section~\ref{sec:a_proofs} of the supplementary material.


\subsection{Threshold selection}
Finding a suitable threshold is crucial in driving accurate tail inference in threshold exceedance modeling, which involves bias and variance trade-offs in the estimation \citep{coles2001introduction, Murphy17122024}.
In the univariate setting, a common approach is to identify a quantity, such as the shape parameter $\xi$ of the GPD, that remains constant across suitable thresholds and use a diagnostic plot to choose a threshold where this quantity stabilizes.
In the multivariate setting, we could exploit Proposition \ref{prop: chi_mGPD}, which states that $\chi_{1,2}(q)$ remains constant for sufficiently large $q$.
But instead of focusing on $\chi_{1,2}(q)$, we follow~\cite{kiriliouk2019peaks} and generalized $\chi_{1,2}(q)$ to dimension $d > 2$.
Similar to Proposition \ref{prop: chi_mGPD}, let $F_j$ be the marginal cdf of a $d$-dimensional mGPD, and $q^* = \inf\{q\in(0,1):F_j^{-1}(q)>0 \text{ for } 
\forall j = 1,\cdots,d\}$, define $\chi_{1:d}(q)$  as 
\begin{equation*}
\begin{aligned}
\chi_{1:d}(q) &= \frac{\mathbb{P}\{\bigcap_{j=1}^d \{X_j>F_j^{-1}(q)\} \}}{1-q},
\end{aligned}
\end{equation*}
and $\chi_{1:d} = \lim_{q \rightarrow 1^-}\chi_{1:d}(q)$, then for any $q^*<q<1$, , we have
\begin{equation}
\begin{aligned}
\chi_{1:d}(q\mid q>q^*) = \chi_{1:d}.
\end{aligned}
\label{eq: chi_1:d}
\end{equation}
To determine an appropriate threshold, \cite{kiriliouk2019peaks} suggested calculating the empirical $\chi_{1:d}(q)$ from the data and using a diagnostic plot to define the threshold as the minimal marginal $q$-quantile such that $\widehat{\chi}_{1:d}(q) \approx \chi_{1:d}$ holds for $q>q^*$.
There are two concerns with Kiriliouk's method. 
Firstly, $\widehat{\chi}_{1:d}(q)$ inevitably goes to zero even for a relatively small $q$ in high dimensions due to sparsity in the joint exceedance region.
Hence, this approach is less effective in high-dimensional applications.
Secondly, $\widehat{\chi}_{1:d}(q)$ only considers the joint exceedance probability. 
When probabilities of partial threshold exceedances are of interest, $\chi_{1:d}(q)$ fails to evaluate the partial lower tail region, i.e., the region where some components fall below the threshold while others exceed it.

We therefore consider another quantity $\omega_{1:d}(q)$, defined as
\begin{equation*}
\begin{aligned}
\omega_{1:d}(q) &= \frac{\mathbb{P}\{\bigcup_{j=1}^d \{ X_j>F_j^{-1}(q)\} \}}{1-q}.
\end{aligned}
\end{equation*}
Let $\omega_{1:d} = \lim_{q \rightarrow1^-}\omega_{1:d}(q)$ and $q^*$ be the same as in \eqref{eq: chi_1:d}, we have 
\begin{equation*}
\begin{aligned}
\omega_{1:d}(q \mid q>q^*) = \omega_{1:d}.
\end{aligned}
\end{equation*}
Similar to $\chi_{1:d}(q)$ and $\chi_{1:d}$, $\omega_{1:d}(q)$ is constant when $q>q^*$ and hence equal to its limit $\omega_{1:d}$.
The advantage of using $\omega_{1:d}(q)$ and $\omega_{1:d}$ is that they properly account for a suitable description of the lower tail margins of the mGPD, which is crucial if partial exceedances are of interest. 
Additionally, $\omega_{1:d}(q)$ accounts for all combinations of threshold exceedances,  hence its empirical estimation is more stable than the empirical $\chi_{1:d}(q)$ due to the relatively large number of exceedances in the numerator, especially when the tail dependence of the data is weak.
The following proposition summarizes the above results and further establishes the form of $\chi_{1:d}$ and $\omega_{1:d}$ for an mGPD with random vector $\boldsymbol{T}$ as the generator.
Proofs are deferred to Section~\ref{sec:a_proofs} of the supplementary material.
\begin{prop}
    Let  $F_j$ be the cdf of an mGPD, $q^* = \inf\{q\in(0,1):F_j^{-1}(q)>0 \text{ for } \forall j = 1,\cdots,d\}$, then the following holds for $\chi_{1:d}$, $\chi_{1:d}(q)$, $\omega_{1:d}$, $\omega_{1:d}(q)$ defined previously:
    \begin{equation}
\begin{aligned}
\chi_{1:d}(q \mid q>q^*) = \chi_{1:d}
&= \mathbb{E}\left(\min \left\{\frac{\exp\{T_1-\max(\boldsymbol{T})\}}{\mathbb{E}(\exp\{T_1-\max(\boldsymbol{T})\})},\cdots,\frac{\exp\{T_d-\max(\boldsymbol{T})\}}{\mathbb{E}(\exp\{T_d-\max(\boldsymbol{T})\})}\right\}\right),\\
\omega_{1:d}(q \mid q>q^*)=\omega_{1:d}
&= \mathbb{E}\left(\max\left\{\frac{\exp\{T_1-\max(\boldsymbol{T})\}}{\mathbb{E}(\exp\{T_1-\max(\boldsymbol{T})\})},\cdots,\frac{\exp\{T_d-\max(\boldsymbol{T})\}}{\mathbb{E}(\exp\{T_d-\max(\boldsymbol{T})\})}\right\}\right).
\end{aligned}
\end{equation}
\label{prop: chi_omega}
\end{prop}
In practice, one can use both $\chi_{1:d}(q)$ and $\omega_{1:d}(q)$ to find two appropriate thresholds and take the component-wise maxima as the final threshold. 
This approach is sketched in Algorithm \eqref{alg: thres_selection}.

\begin{algorithm}
\caption{\textbf{Threshold Selection Method for GPDFlow}}\label{alg: thres_selection}
\vspace{0.1cm}
\hspace{0.3cm} \textbf{Input}: All raw observations (rather than only the threshold‐exceedance subset) \\   \hspace*{1.8cm} $ \boldsymbol{y}_i=(y_{i1},\cdots,y_{id}), \quad i = 1,\cdots,N$
\begin{algorithmic}
\State 1. Calculate the empirical marginal cdf $\hat{F}_j(y_j) = \sum_{i=1}^N \mathbbm{1}\{y_{ij}<y_j\}/N $ for $j=1,\cdots,d$
\State 2. For $q \in (0,1)$, calculate the empirical value of 
    \Indent
          \State $\widehat{\chi}_{1:d}(q) = \frac{ \sum_{i=1}^N \mathbbm{1}\{\bigcap_{j=1}^d \{\widehat{F}_j(Y_j)>q\}\}}{N(1-q)}$
          \State $\widehat{\omega}_{1:d}(q) = \frac{ \sum_{i=1}^N \mathbbm{1}\{\bigcup_{j=1}^d \{ \widehat{F}_j(Y_j)>q\}\}}{N(1-q)}$
    \EndIndent
\State 3. Find the minimal $q_{\chi}$ such that $\widehat{\chi}_{1:d}(q)$ is approximately constant for $q>q_{\chi}$
\State 4. Find the minimal $q_{\omega}$ such that $\widehat{\omega}_{1:d}(q)$ is approximately constant for $q>q_{\omega}$
\State 5. Set $q^* = \max\{(q_{\chi},q_{\omega})\}$. The threshold is given by $(\hat{F}_1^{-1}(q^*),\cdots, \hat{F}_d^{-1}(q^*))$\\
\textbf{Output}: $d$-dimensional threshold vector $(\hat{F}_1^{-1}(q^*),\cdots, \hat{F}_d^{-1}(q^*))$
\end{algorithmic}
\end{algorithm}

\section{Simulation}
\label{sec: simulation}
We design two simulation scenarios to assess the estimation of GPDFlow. 
The first one is a well-specified scenario, where data are simulated from a traditional parametric mGPD. 
The purpose is to evaluate if GPDFlow can be correctly estimated.
Specifically, we assess dependence and marginal estimation by examining all pairwise bivariate tail dependence coefficient estimates and marginal parameters $\boldsymbol{\sigma}$ and $\boldsymbol{\gamma}$.
The second scenario is misspecified to test the model's robustness, with data drawn from a distribution constructed using the copula approach.
Here, we compare the estimated density of threshold exceedance data and the joint exceedance probability.

In the first scenario, we evaluate the goodness of fit of GPDFlow in dimensions $d=2,3,5$. 
The data are simulated from a parametric mGPD with $\boldsymbol{\sigma} = (0.5, 1.2, 1, 1.5, 0.8)$, $\boldsymbol{\gamma} = (-0.1, 0.2, 0, 0.15, -0.05)$, and  $f_{\boldsymbol{T}}(\boldsymbol{t}) = \prod_{j=1}^d \exp\{t_j+\beta_j\}/a_j$, where $\boldsymbol{a}=(2,0.5,1, 5,1.5)$ and $\boldsymbol{\beta} = (1, 2,3,4,5)$.
For $d=$ 2 or 3, the parameters are restricted to the first $d$ elements of each vector.
The flow model for $f_{\boldsymbol{T}}$ consists of 16 layers, where the log-scale and translation in each layer are MLPs with one hidden layer and $4d$ hidden neurons. 
For each dimension, we generate 100 samples from the parametric mGPD above and fit a GPDFlow using the specified flow architecture to assess robustness under sparse data conditions. 
Given the small sample size, we train GPDFlow with a relatively large epoch (200) to ensure a comprehensive update of the flow parameters. 
The total runtime for this model is approximately 2 minutes on an NVIDIA RTX A6000.

We repeat the experiment 100 times for each dimension, and Figure \ref{fig: simulation_1} displays a comparison between the theoretical values of pairwise tail dependence, marginal parameters, and their corresponding Monte Carlo estimates. 
The closed-form of the theoretical pairwise tail dependence coefficient of the given $f_{\boldsymbol{T}}$ in this simulation scenario can be found in the supporting materials of \cite{kiriliouk2019peaks}.
Overall, the estimation performance is satisfactory, with most theoretical values falling within the interquartile range of the GPDFlow estimates. 
The model effectively captures both strong tail dependence ($\chi \approx 0.72$) and weak tail dependence  ($\chi \approx 0.35$) present in the same dataset ($d=5$).
Additionally, GPDFlow accurately estimates the values of $\gamma_j$, correctly identifying whether each margin has a heavy, Pareto-like tail ($\gamma_j>0$),  a short or bounded tail ($\gamma_j<0$), or a light, exponential tail ($\gamma_j=0$).
Notably, the estimations remain robust across dimensions, as we do not observe increasing uncertainty or significant bias as $d$ grows.
\begin{figure}[H]
\begin{center}
\includegraphics[scale=0.28]{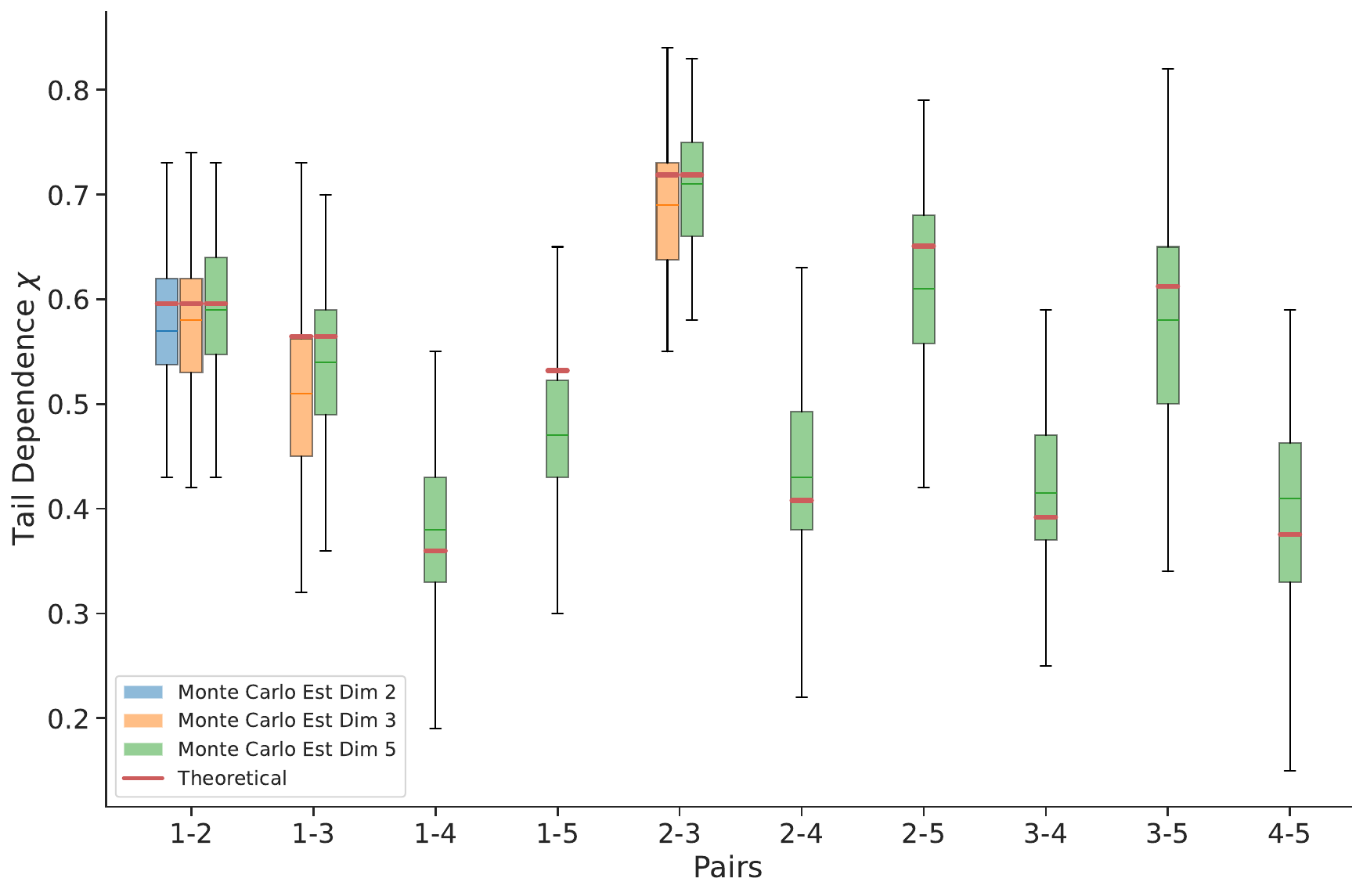}\\
\includegraphics[scale=0.28]{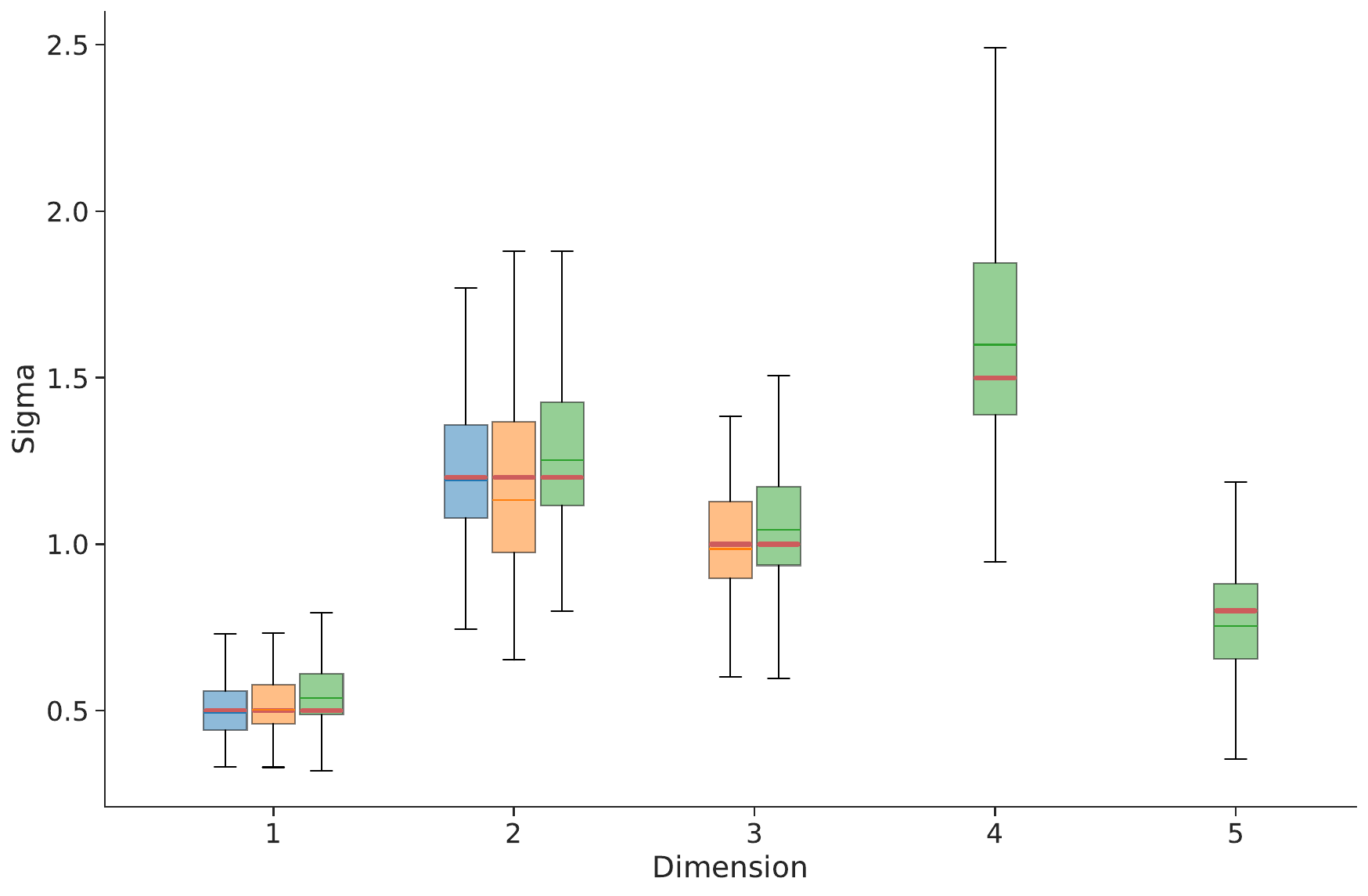}
\includegraphics[scale=0.28]{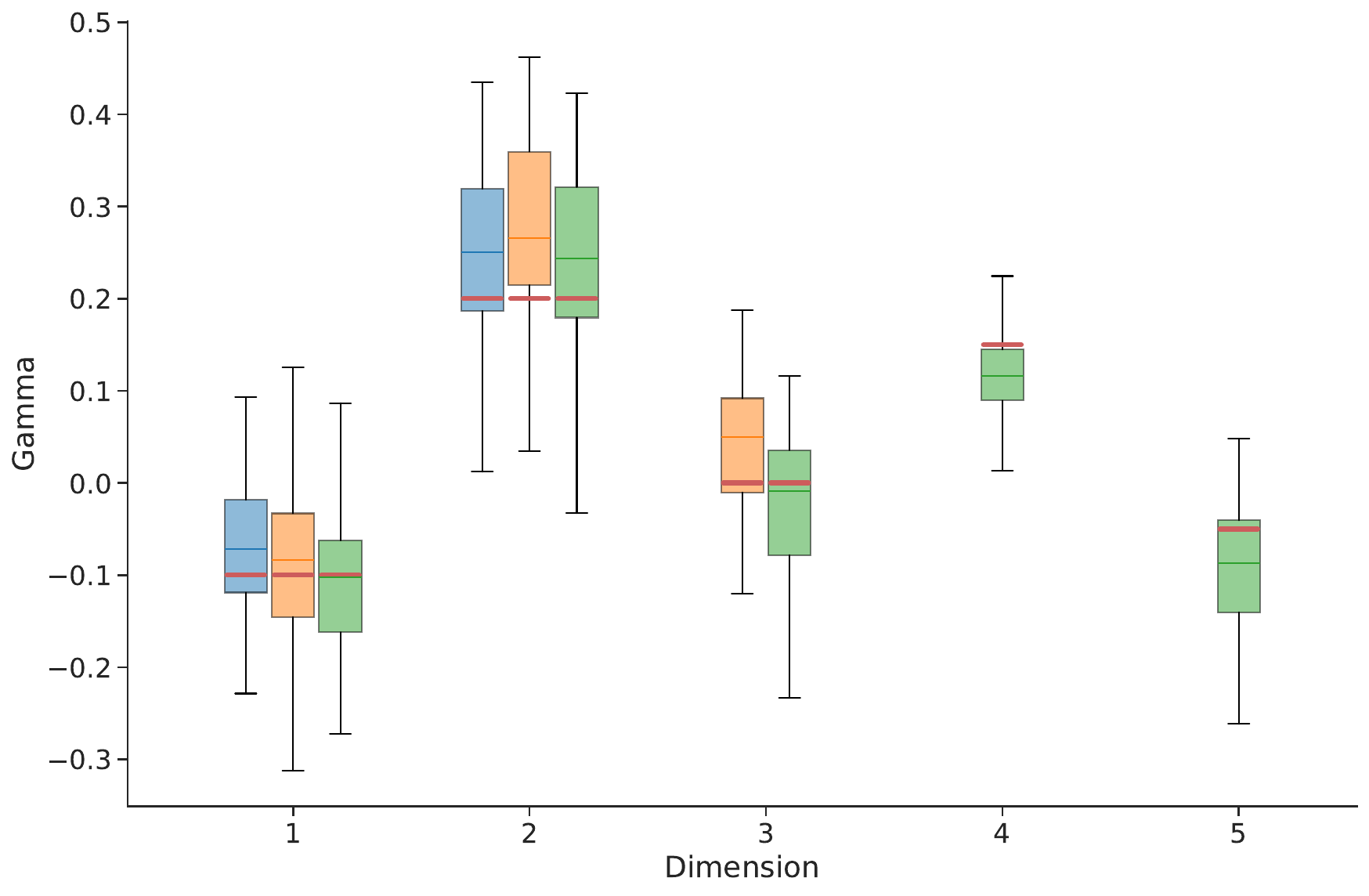}
\end{center}
\caption{\footnotesize{Boxplots of estimated pairwise $\chi$ (Top) and marginal parameters $\boldsymbol{\sigma}$, $\boldsymbol{\gamma}$ (Bottom two) across dimension $d=$  2, 3, and 5.
The true model is a mGPD with
$\boldsymbol{\sigma} = (0.5, 1.2, 1, 1.5, 0.8)$, $\boldsymbol{\gamma} = (-0.1, 0.2, 0, 0.15, -0.05)$, and $f_{\boldsymbol{T}}(\boldsymbol{t}) = \prod_{j=1}^d \exp\{t_j+\beta_j\}/a_j$, where $\boldsymbol{a}=(2, 0.5, 1, 5, 1.5)$ and $\boldsymbol{\beta} = (1, 2,3,4,5)$.
Point estimates for $\boldsymbol{\sigma}$ and $\boldsymbol{\gamma}$ are directly obtained from the GPDFlow estimation based on 100 simulations, while the point estimates for $\chi$ are based on the empirical values derived from 10,000 predictions in each of the 100 simulations.
}}
\label{fig: simulation_1}
\end{figure}
In the second scenario, we assess GPDFlow's estimation performance when data are not in the mGPD framework. 
To facilitate visualization, simulations are conducted in two dimensions.
The data are generated from a bivariate Gumbel copula $C_{\text{Gumbel}}(\cdot;\theta)$ with parameter $\theta=1.3$, and two Gaussian margins with location and scale parameters $(\mu_1,s_1)= (1,3)$ and  $(\mu_2,s_2)=(2,5)$.
In other words, the data comes from the distribution
\begin{align}\label{eq:sim2dist}
    F_{Y_1,Y_2}(y_1,y_2)= C_{\text{Gumbel}}\{\Phi[s_1^{-1}(y_1-\mu_1)], \Phi[s_2^{-1}(y_2-\mu_2)];\theta\},
\end{align}
where $C_{\text{Gumbel}}(u,v;\theta)=\exp\{-[\left(-\log u\right)^{\theta}+(-\log v)^{\theta}]^{1/\theta}\}$, and $\Phi$ is the cdf of a standard Gaussian distribution.
We generate 1200 samples, and Algorithm \ref{alg: thres_selection} identifies $q=0.95$ as a suitable threshold, yielding approximately 100 threshold exceedance samples. 
These exceedance samples are then fitted using a GPDFlow model with 16 layers of transformations, each consisting of MLPs with one hidden layer and 20 hidden neurons.
To assess estimation accuracy, we repeat this procedure 100 times and compare the Monte Carlo mean of the GPDFlow estimates with the corresponding theoretical values.

The top two plots in Figure \ref{fig: simulation_2} compare the theoretical and GPDFlow-estimated densities of the threshold exceedance data.
The theoretical density 
is obtained as 
\begin{align}\label{eq: sim2dens}
    f_{X_1,X_2}(x_1,x_2) = \frac{f_{Y_1,Y_2}(x_1+\tau_1,x_2 + \tau_2)}{1-C_{\text{Gumbel}}(0.95,0.95;\theta)}\mathbbm{1}\left\{\bigcup_{j=1}^2  \left\{\Phi \left(\frac{x_j+\tau_j-\mu_j}{s_j}\right)>0.95 \right\}\right\},
\end{align}
where $\boldsymbol{\tau}=(\tau_1,\tau_2)$ is the threshold, and $\tau_j$ is the 0.95-quantile of margin $j$, $j=1,2$ .
\eqref{eq: sim2dens} corresponds to a truncated $f_{Y_1,Y_2}(y_1,y_2)$ with support only in the region where at least one component exceeds the threshold. To facilitate the visualization, we shift this support by subtracting the threshold, and positioning the new boundary along negative x- and y-coordinates. 
The contour plots show that GPDFlow generally provides a good approximation of the theoretical density, particularly in high-density regions. 
Minor discrepancies appear in the lower tail and areas where the density is close to zero.
These differences arise because GPDFlow defines a bounded distribution with support depending on $\boldsymbol{\gamma}$, whereas the theoretical density remains unbounded.

\begin{figure}[!htbp]
\begin{center}
\includegraphics[scale=0.32]{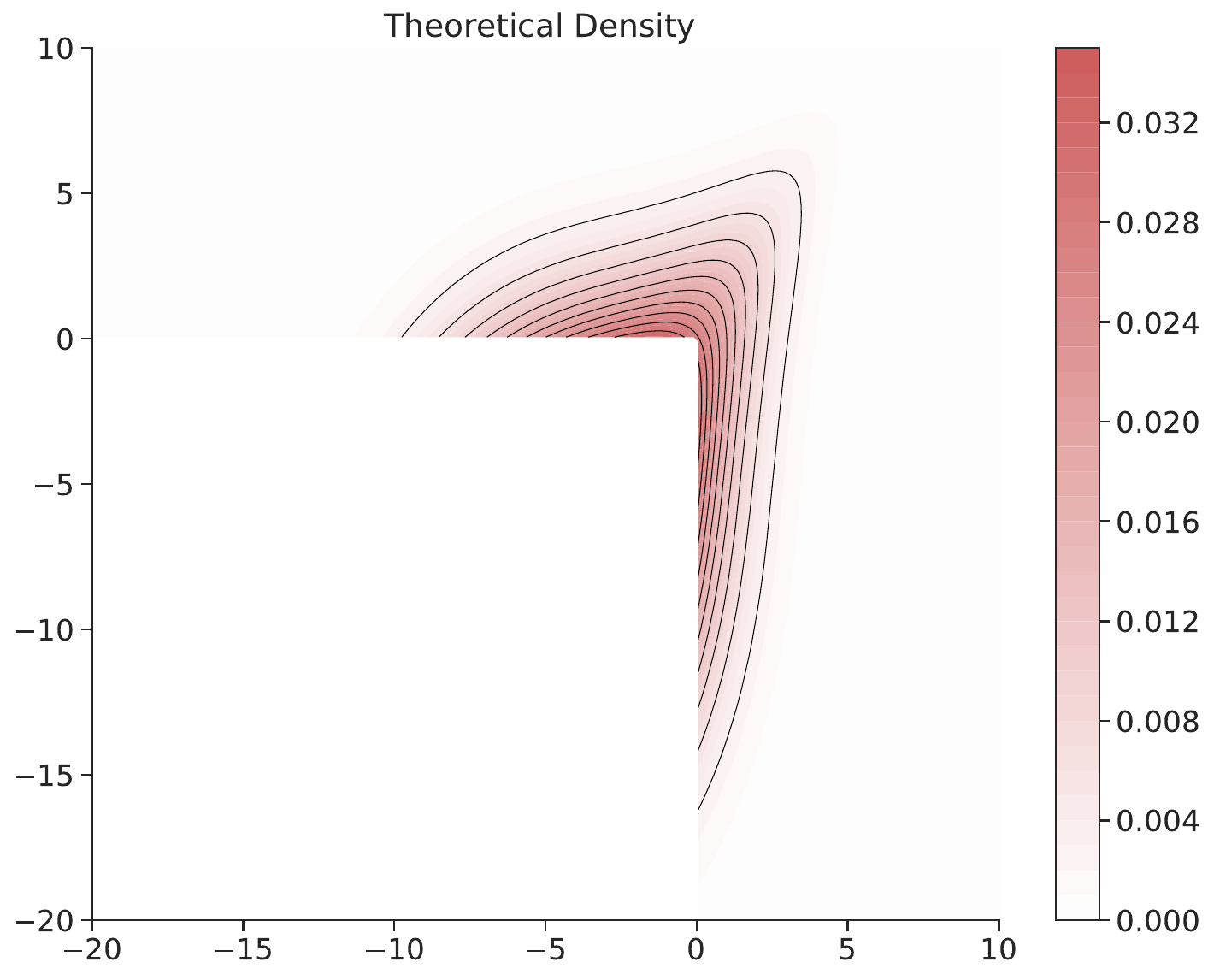}
\includegraphics[scale=0.32]{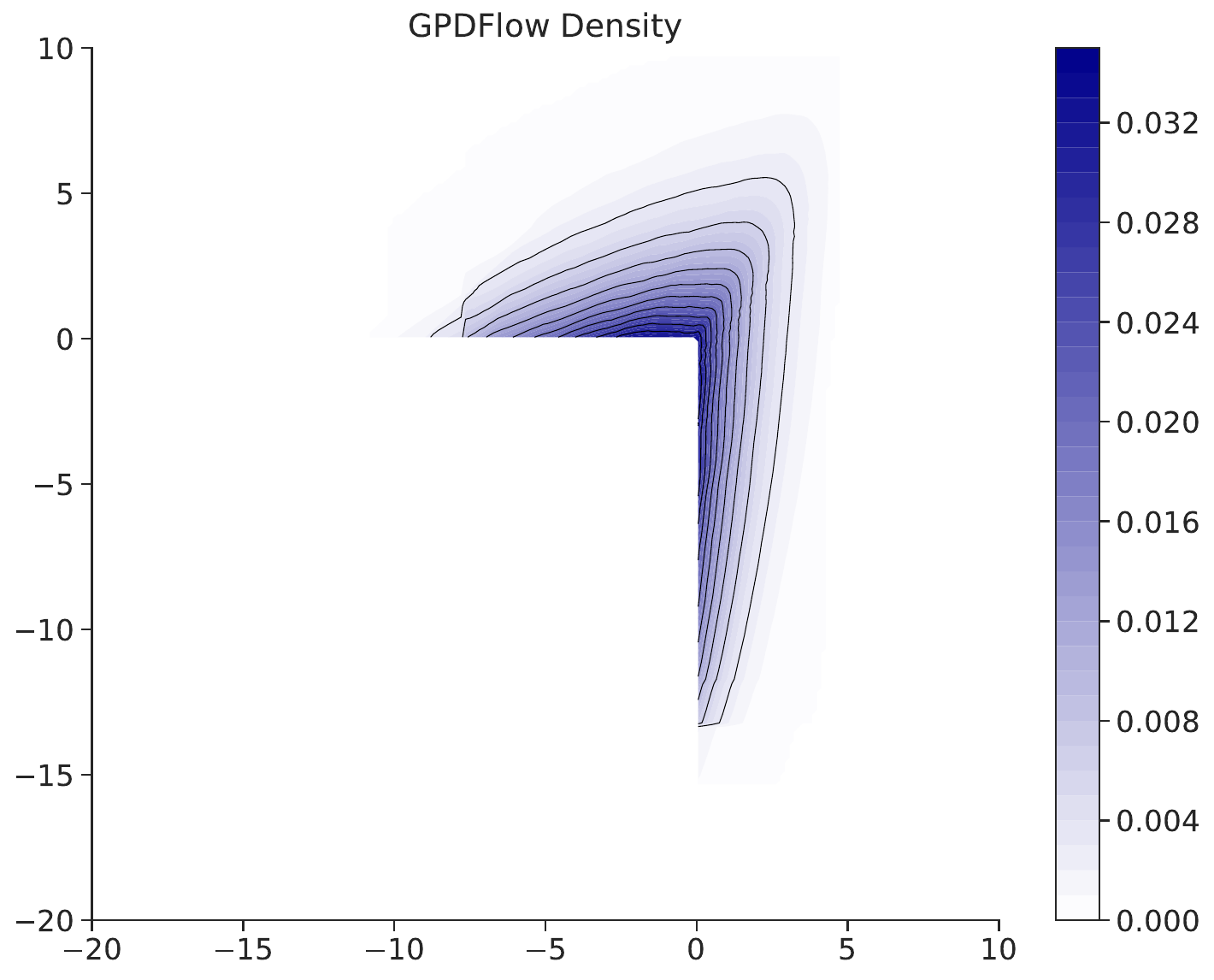}
\includegraphics[scale=0.32]{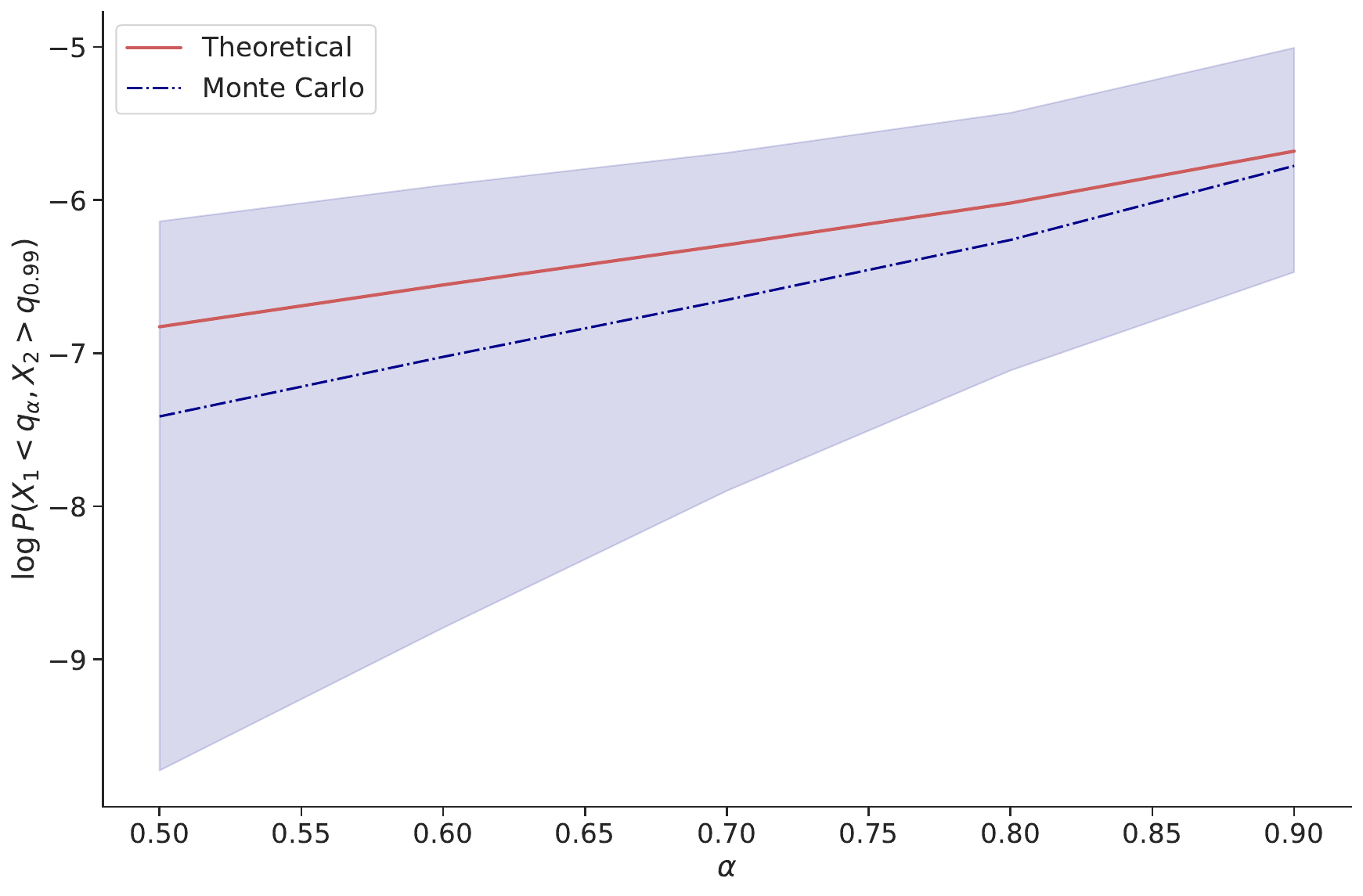}
\end{center}
\caption{\footnotesize{Theoretical and estimated density contour plots of the threshold exceedance data (top panel) and plot of $\log \mathbb{P}(Y_1<q_{1, \alpha}, Y_2>q_{2, 0.99})$ (bottom). The entire data comes from the distribution in~\eqref{eq:sim2dist}.
The theoretical threshold density is obtained following~\eqref{eq: sim2dens}, while the contour of the GPDFlow density is the average density over 100 simulations. 
In the bottom plot, the theoretical line is derived from \eqref{eq: sim2_theoretical_partial_exceedance_prob}, while the dark blue line is the mean of empirical estimation of \eqref{eq: simulation_2} from 100 thousand GPDFlow samples over 100 simulations, and the shaded area covers the 2.5\% and 97.5\% of the empirical estimations.
}}
\label{fig: simulation_2}
\end{figure}
We further scrutinize the performance of GPDFlow by checking the estimation of joint probabilities associated with partial extremes.
Specifically, we analize the estimation of $\mathbb{P}(Y_1<q_{1,\alpha}, Y_2>q_{2,0.99})$, where $Y_1$ and $Y_2$ are the random variables of the simulated data, $q_{1,\alpha}$ is the $\alpha$-quantile of margin 1 and $q_{2,0.99}$ is the 0.99 quantile of margin 2.
The theoretical value of such probability is straightforward to calculate by the cdf of Gumbel Copula. 
Indeed, we have 
\begin{equation}
 \begin{aligned}
     \mathbb{P}(Y_1 < q_{1,\alpha}, Y_2>q_{2,0.99}) &= \mathbb{P}(Y_1< q_{1,\alpha})-\mathbb{P}(Y_1< q_{1,\alpha}, Y_2< q_{2,0.99})\\
&= \alpha - C_{\text{Gumbel}}(\alpha,0.99;\theta).
 \end{aligned}
\label{eq: sim2_theoretical_partial_exceedance_prob}
\end{equation}

To estimate $\mathbb{P}(Y_1 < q_{\alpha}, Y_2>q_{0.99})$ by GPDFlow, we can write 
\begin{equation}
    \begin{aligned}
&\mathbb{P}(Y_1 < q_{1,\alpha}, Y_2>q_{2,0.99}) \\=&\mathbb{P}\left(Y_1 < q_{1,\alpha}, Y_2>q_{2,0.99}, \bigcup_{j=1}^2\{ Y_j > q_{j,0.95}\}\right) \\
=&\mathbb{P}\left(Y_1 - \tau_1 < q_{1,\alpha} -\tau_1, Y_2 - \tau_2>q_{2,0.99}-\tau_2 \bigg\rvert \bigcup_{j=1}^2\{ Y_j -\tau_j > 0\}\right)\mathbb{P}\left( \bigcup_{j=1}^2\{ Y_j > \tau_j\}\right).
\end{aligned}
\label{eq: simulation_2}
\end{equation}
The first term in the last equation of \eqref{eq: simulation_2} is a probability conditioned on at least one extreme component, making it well-suited for estimation via GPDFlow. 
The second term is a threshold exceedance probability of a moderately high threshold, which can be well estimated empirically.
The bottom plot in Figure \ref{fig: simulation_2} compared the theoretical probability $\mathbb{P}(Y_1<q_{1, \alpha}, Y_2>q_{2, 0.99})$ with its Monte Carlo estimate from GPDFlow for $\alpha\in (0.5,0.9)$ on a logarithm scale.
The GPDFlow estimates closely align with the theoretical value, though a slight underestimation can be observed across all $\alpha$.
As $\alpha$ increases, both bias and variance decrease because the data become more extreme and thus more closely match the theoretical limiting distribution, allowing GPDFlow to characterize them more accurately..

\section{Application}
\label{sec: application}
We apply our GPDFlow to multivariate risk analysis of five major US banks.
The key risk measure in this study is Conditional Value-at-Risk (CoVaR), a pairwise Value-at-Risk that accounts for the dependence between two financial entities, such as institutions or portfolios.
While the definition of CoVaR may vary, it generally falls within a framework that quantifies the risk of one entity given that one or more others are in distress.
For instance, \cite{mainik2014dependence} defines the CoVaR of two random variable $X$ and $Y$ as 
\begin{align}
    \text{CoVaR}_{\alpha,\beta}(Y|X) = \text{VaR}_{\alpha}\{Y|X\geq \text{VaR}_{\beta}(X)\},
    \label{eq: CoVaR}
\end{align}
where $\text{VaR}_{\eta}(Z) = \inf\{ z \in \mathbb{R}: F_Z(z) \geq \eta \}$, $\eta\in (0,1)$, and $F_Z$ is the cdf of random variable $Z$.
This definition extends beyond individual risks by providing a method to measure systemic risk.
Estimating CoVaR requires the entire marginal distribution conditioned on the extreme behavior of the other component, a task that can be naturally handled by GPDFlow. 

\begin{figure}[H]
\begin{center}
\includegraphics[width=3.5in]{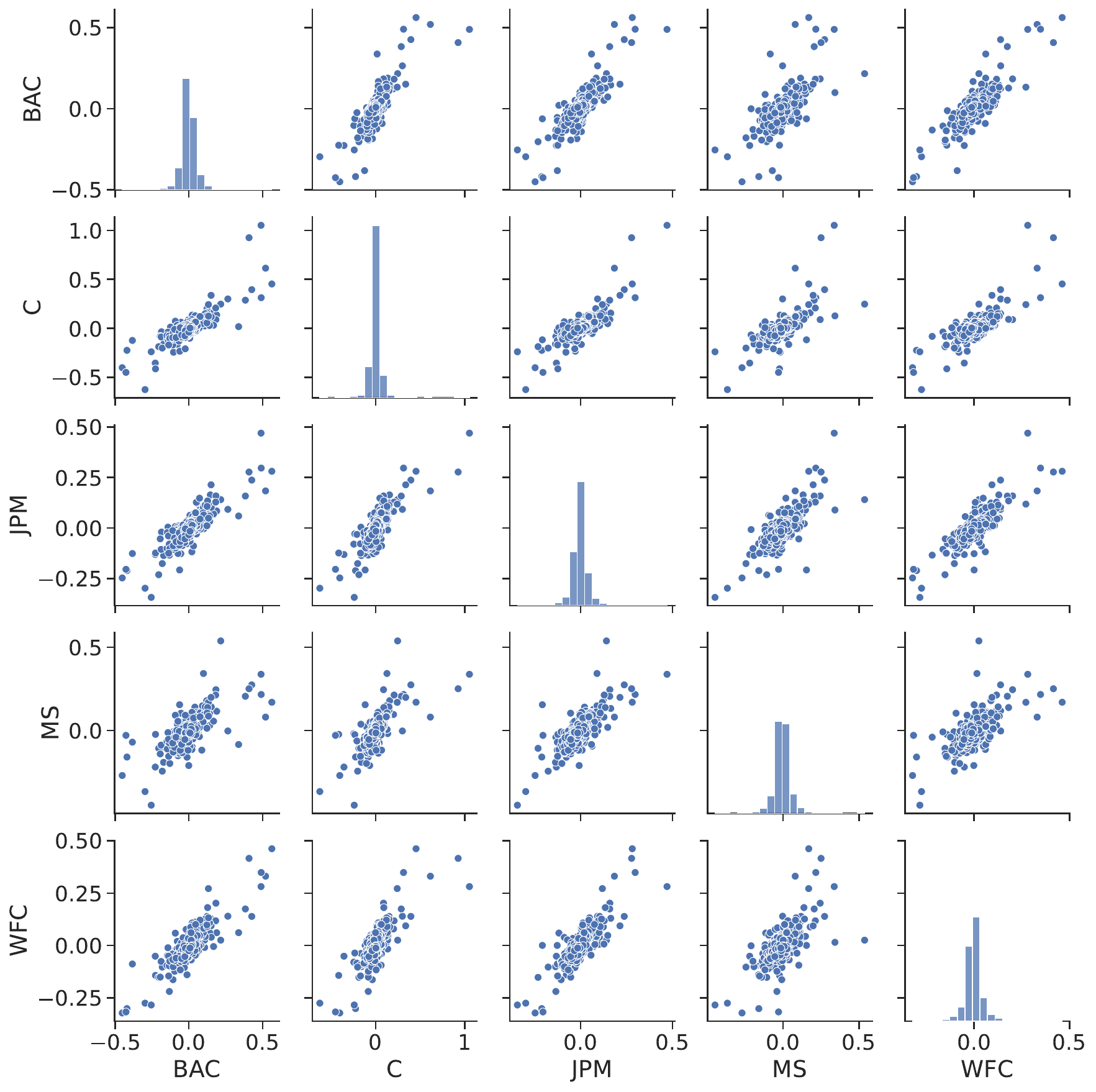}
\includegraphics[width=3in]{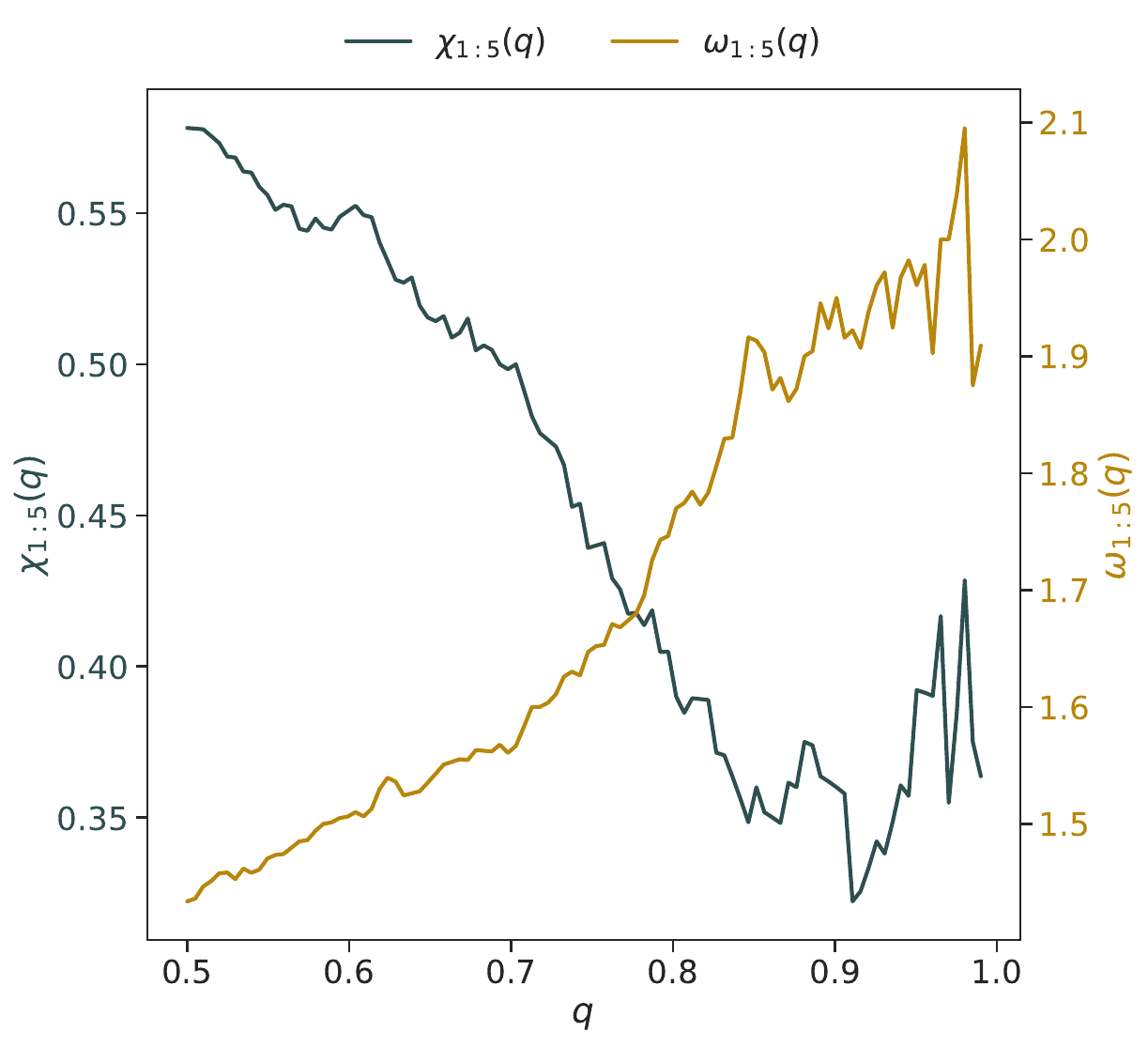}
\end{center}
\caption{\footnotesize{Scatter plots of the negative log-return (left) of the five US banks and empirical $\chi_{1:d}(q)$ and $\omega_{1:d}(q)$ of the negative log-return calculated by Algorithm \ref{alg: thres_selection}. }}
\label{fig: scatter_thres_selection}
\end{figure}
We demonstrate the application of GPDFlow in risk management by examining the CoVaR of the following five large US banks: JPMorgan Chase (JPM), Bank of America (BAC), Citigroup (C), Wells Fargo (WFC) and Morgan Stanley (MS).
Using 5-day closing prices from these banks between January 1, 2005, and February 1, 2025, sourced from Yahoo Finance, we examine the negative log returns, which total 1,010 observations. Since our focus is on general risk rather than forecasting, we do not remove heteroscedasticity from the negative log returns, which is the same treatment as in the analysis of the negative return of bank stocks in \cite{kiriliouk2019peaks}. 
Figure \ref{fig: scatter_thres_selection} shows scatter plots of the negative log return as well as plots of the two tail dependence measures introduced in Proposition \ref{prop: chi_omega} for the five banks. 
Clear asymptotic dependence is observed in the scatter plot, which is confirmed by the positive values of $\widehat{\chi}_{1:5}(q)$ on the right.
The empirical $\omega_{1:5}(q)$ seems to enter a plateau at around $q=0.9$, while this happens at about $q=0.95$ for $\chi_{1:5}(q)$.
This suggests that the tail data can be well approximated by GPDFlow when at least one component exceeds its 0.95-quantile. 
Consequently, we select the marginal 0.95-quantile as the threshold and model the data above it (100 observations total) using GPDFlow.

For comparison purposes, we also implement a parametric mGPD from \cite{kiriliouk2019peaks} on our data and compare its performance with that of 
GPDFlow.
This parametric mGPD is chosen from three mGPD models, where $f_{\boldsymbol{T}}$ is either independent reverse-exponential, independent Gumbel, or multivariate Gaussian, following the selection procedure in the original paper. The resulting model is given by $f_{\boldsymbol{T}}(t_1,\cdots,t_5)=\prod_{j=1}^5 \alpha_j\exp(-\alpha_j t_j)\exp[-\exp(-\alpha_j t_j)], \; \alpha_j>0,$ with free parameters $\boldsymbol{\sigma}$ and $\boldsymbol{\gamma}$. Parameters in this model are estimated via a censored-likelihood method, in which components below the threshold are treated as censored.

Figure \ref{fig: model_comparison} compares the performance of GPDFlow and the parametric mGPD in estimating both the marginal densities and tail dependence of the five banks’ threshold exceedance data.
For the marginal density of JPMorgan Chase and Morgan Stanley (with others provided in Section~\ref{sec: appendix_fig5} of the Appendix), GPDFlow provides a strong approximation across the lower tail, upper tail, and high-density regions, without overfitting to the wiggly empirical right tails.
In contrast, although the parametric mGPD does a relatively good job estimating the right tails, it significantly deviates from the empirical density in the lower-tail and high-density regions of the two banks.

In the tail dependence plots, both the estimated $\chi_{1:5}^m(q)$, $\omega_{1:5}^m(q)$, $m \in \{\text{GPDFlow, mGPD}\}$ from GPDFlow and the parametric mGPD show constant behavior for $q<0.95$, consistent with Proposition \ref{prop: chi_omega}. 
This constant trend theoretically extends to the region $q>0.95$ but is subject to more volatility due to the limited number of predicted observations in the simulated data.
Comparing the empirical $\chi_{1:5}$ and $\omega_{1:5}$ with the estimated $\chi_{1:5}^m$ and $\omega^m_{1:5}$ (approximately by $\chi_{1:5}^m(q \mid 0.5<q<0.95)$ and $\omega_{1:5}^m(q \mid 0.5<q<0.95)$, respectively), it is clear that the parametric mGPD underestimates $\chi_{1:5}$ and overestimates $\omega_{1:5}$. Its estimates lie near the boundary of the 95\% bootstrap confidence interval of the empirical values. 
In contrast, GPDFlow estimations exhibit a much smaller bias, with both $\chi^{\text{GPDFlow}}_{1:5}$ and $\omega_{1:5}^{\text{GPDFlow}}$ falling well within the central region of the corresponding empirical confidence interval.
\begin{figure}[H]
\begin{center}
\includegraphics[width=2.5in]{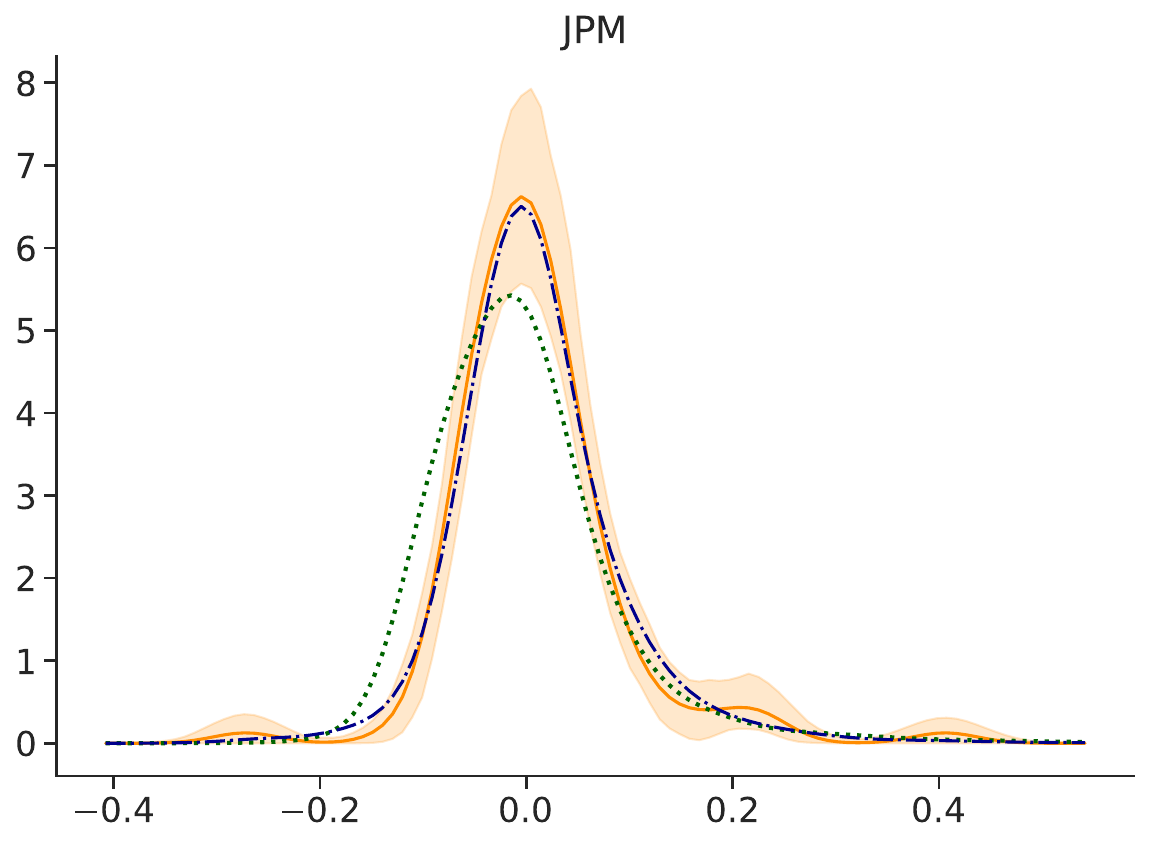}
\includegraphics[width=2.5in]{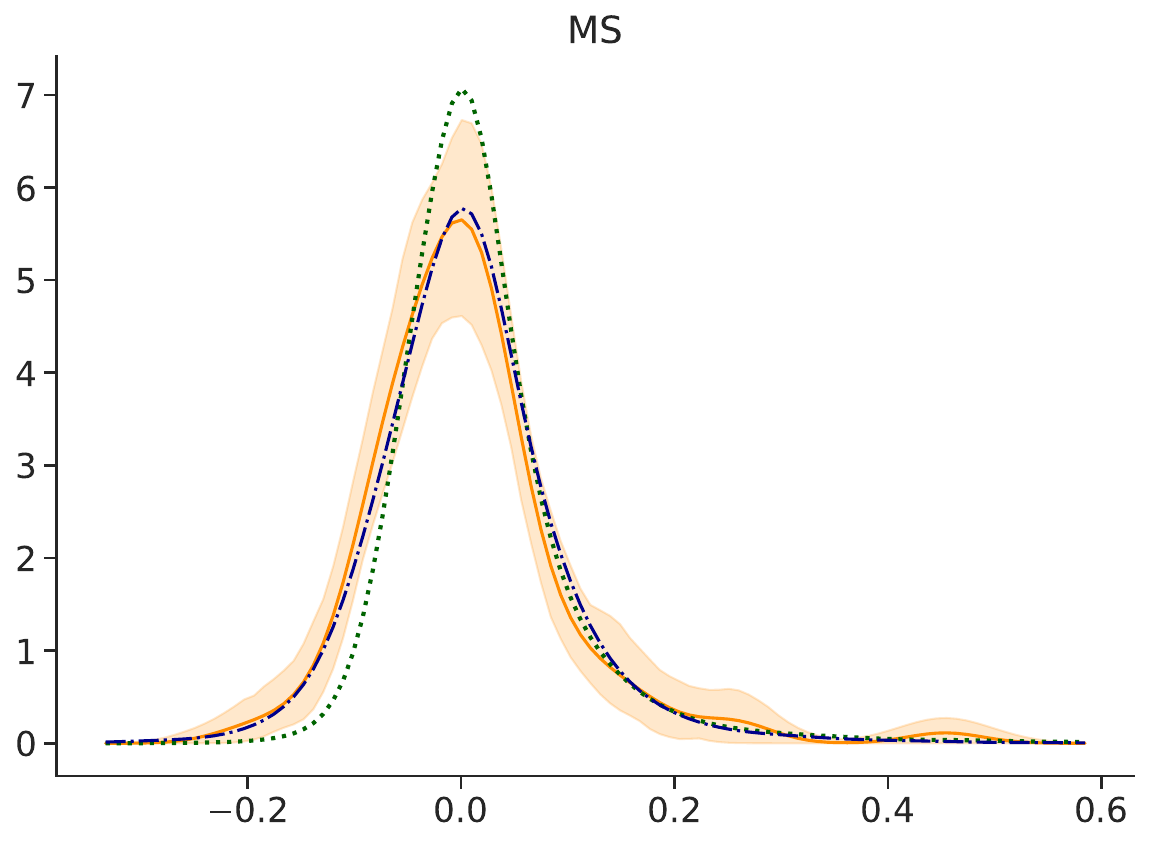}
\includegraphics[width=2.5in]{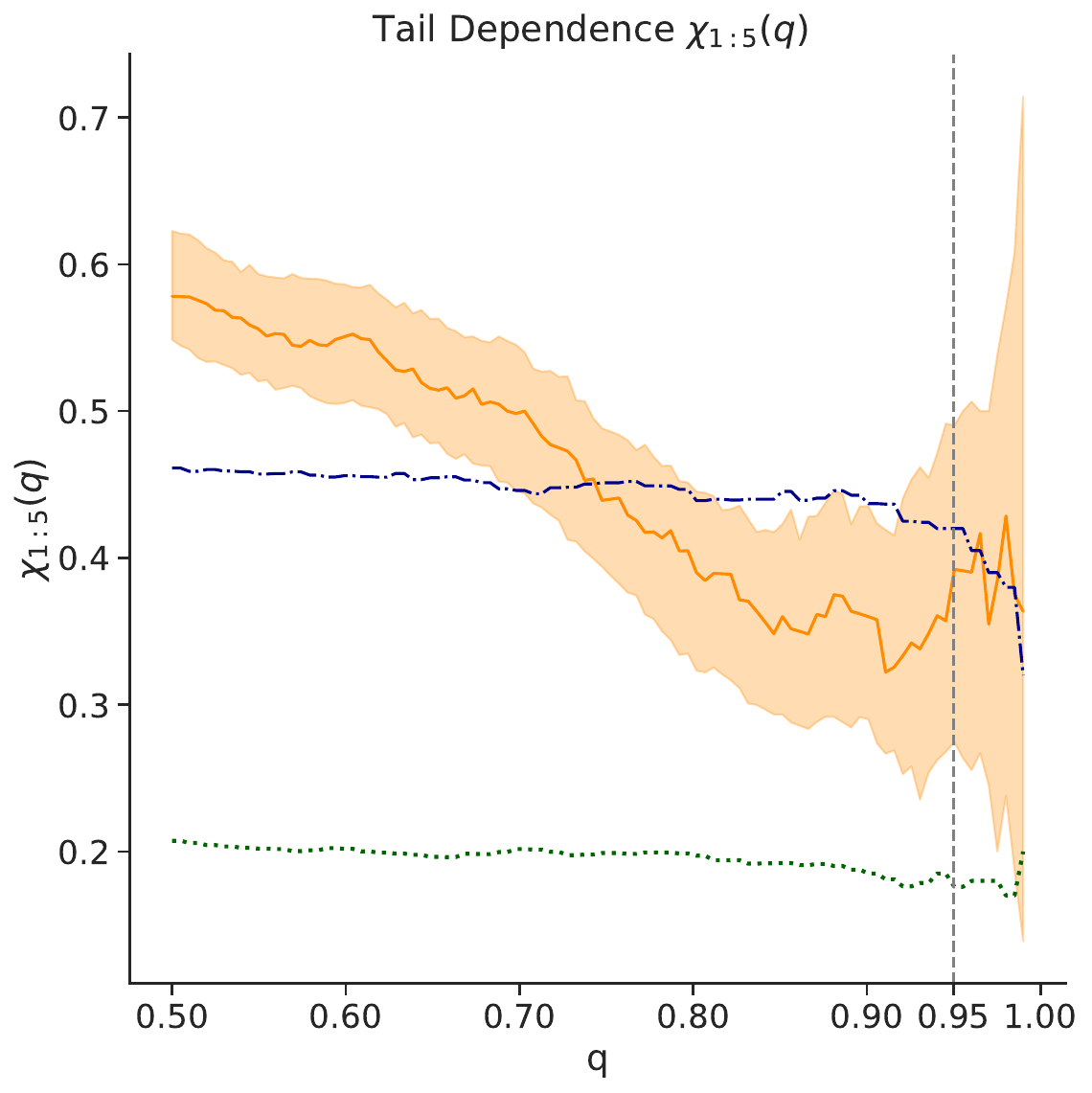}
\includegraphics[width=2.5in]{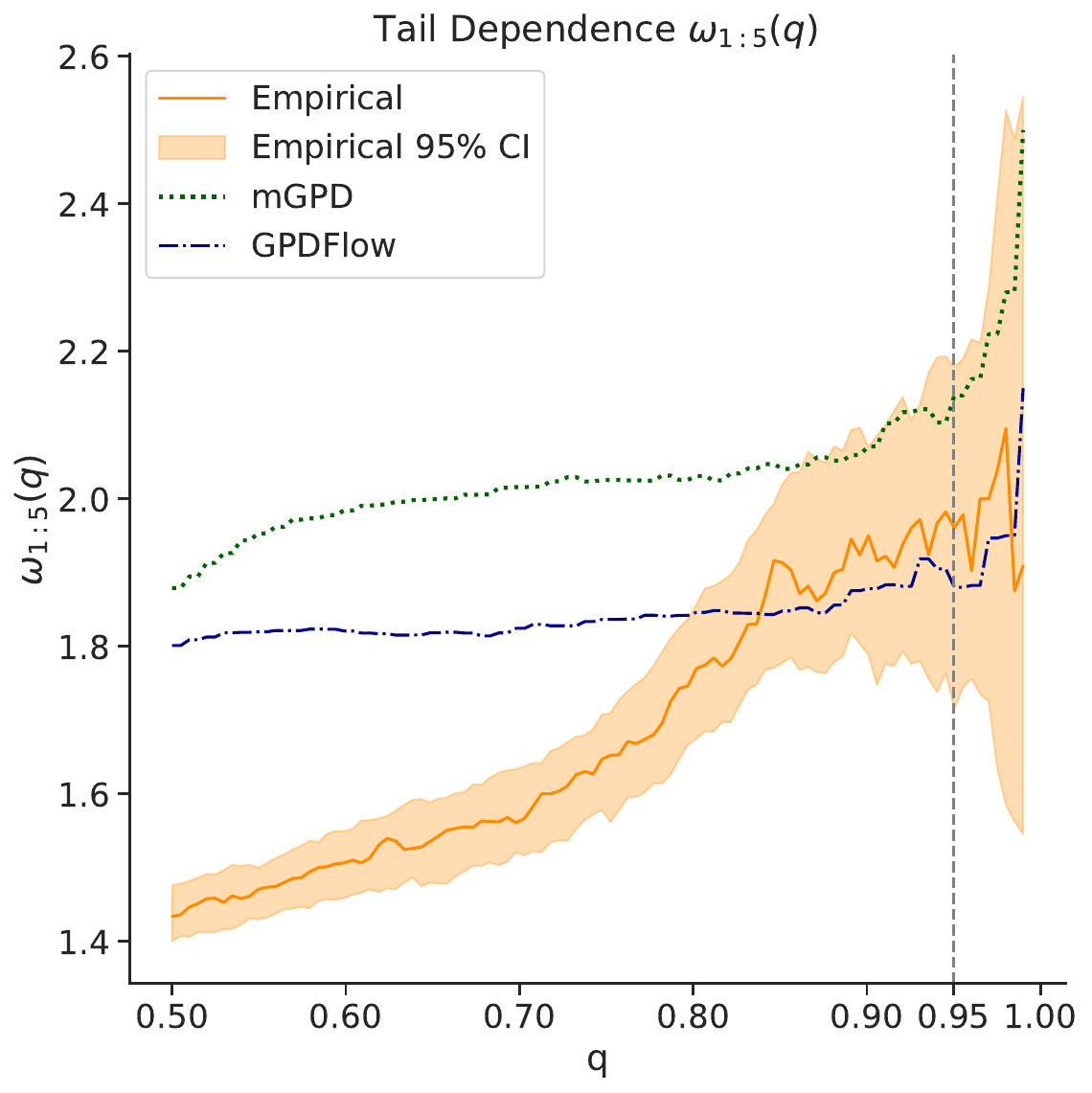}
\end{center}
\caption{\footnotesize{Plots of marginal densities of the threshold exceedance data (top panel) and tail dependence measures (bottom panel). The 95\% CI for the empirical density and dependence measures are generated via Bootstrap. The estimates from the mGPD and GPDFlow are averages derived from 100 Monte Carlo sample sets, each of the same size as the threshold exceedance data and simulated using mGPD or GPDFlow.
}}
\label{fig: model_comparison}
\end{figure}

Next, we use the GPDFlow predictions to estimate the CoVaR between the five banks.
We illustrate the CoVaR that conditions on the largest bank, JPM, under a stress scenario (specifically, $\beta=0.95$).
Let the negative log-returns of the five banks be denoted by $Y_{\text{JPM}}$, $Y_{\text{BAC}}$, $Y_{\text{C}}$, $Y_{\text{WFC}}$, $Y_{\text{MS}}$ and their corresponding 0.95-quantile threshold by $\tau_{\text{JPM}}$, $\tau_{\text{BAC}}$, $\tau_{\text{C}}$, $\tau_{\text{WFC}}$, $\tau_{\text{MS}}$.
By definition, for any two banks $i,j \in \{\text{BAC, C, JPM, MS, WFC}\},\; i\neq j$, \eqref{eq: CoVaR} can be written as 
\begin{equation}
\begin{aligned}
\alpha  &= \frac{\mathbb{P}\{Y_j < \text{CoVaR}_{\alpha,\beta}(Y_j|Y_i),Y_i> \text{VaR}_{\beta}(Y_i) \}}{\mathbb{P}\{Y_i>\text{VaR}_{\beta}(Y_i) \}} \\
  &= \frac{\mathbb{P}\{Y_j < \text{CoVaR}_{\alpha,\beta}(Y_j|Y_i),Y_i> \text{VaR}_{\beta}(Y_i) ,\bigcup_{i=1}^5\{Y_i>\text{VaR}_\beta(Y_i)\}\}}{\mathbb{P}\{Y_i>\text{VaR}_{\beta}(Y_i), \bigcup_{i=1}^5\{Y_i>\text{VaR}_\beta(Y_i)\} \}} \\
  &= \frac{\mathbb{P}\{Y_j < \text{CoVaR}_{\alpha,\beta}(Y_j|Y_i),Y_i> \text{VaR}_{\beta}(Y_i) |\bigcup_{i=1}^5\{Y_i>\text{VaR}_\beta(Y_i)\}\}}{\mathbb{P}\{Y_i>\text{VaR}_{\beta}(Y_i)| \bigcup_{i=1}^5\{Y_i>\text{VaR}_\beta(Y_i)\} \}} 
\end{aligned}
\label{eq: CoVaR_quantile}
\end{equation}
When $\beta \geq 0.95$,   we have $\text{VaR}_\beta(Y_i) \geq \tau_i$. Consequently, the set $\bigcup_{i=1}^5\{Y_i>\text{VaR}_\beta(Y_i)\}$ is a subregion of the threshold exceedance data $\bigcup_{i=1}^5\{Y_i>\tau_i\}$, which Algorithm \ref{alg: thres_selection} indicates can be well approximated by GPDFlow. 
By setting a threshold $\boldsymbol{\tau}^*$, where $\boldsymbol{\tau}^*_i =\text{VaR}_\beta(Y_i)$, $i\in\{\text{BAC, C, JPM, MS, WFC}\}$, both the numerator and denominator in the last equation of \eqref{eq: CoVaR_quantile} can be approximated by a GPDFlow model under this new threshold.
Therefore, estimating $\text{CoVaR}_{\alpha,\beta}(Y_j|Y_i)$ on the raw negative log-return scale is equivalent to finding the $\text{CoVaR}_{\alpha,\beta}(Y_j-\tau^*_j|Y_i-\tau_i^*)$ on the threshold exceedance scale and then translating it back by $\tau^*_j$.
Figure \ref{fig: CoVaR} displays the CoVaR of BAC, C, MS and WFC conditioned on JMP being in distress (i.e., $Y_{\text{JPM}}> \text{VaR}_{0.95}(Y_{\text{JPM}})$).
The GPDFlow estimates accurately fit the empirical CoVaR across all ranges of $\alpha$, with larger uncertainty for higher $\alpha$. Nevertheless, nearly all empirical CoVaR values fall within the 95\% Monte Carlo confidence interval.
Citigroup exhibits the strongest tail dependence with JPMorgan Chase, with $\text{CoVaR}_{\alpha,0.95}$(C$|$JPM)  being higher for large $\alpha$ compared to the other three banks.
Conversely, Morgan Stanley is less affected by JPMorgan Chase's distress, showing the smallest CoVaR among the four banks.
\begin{figure}[H]
\begin{center}
\includegraphics[width=2.5in]{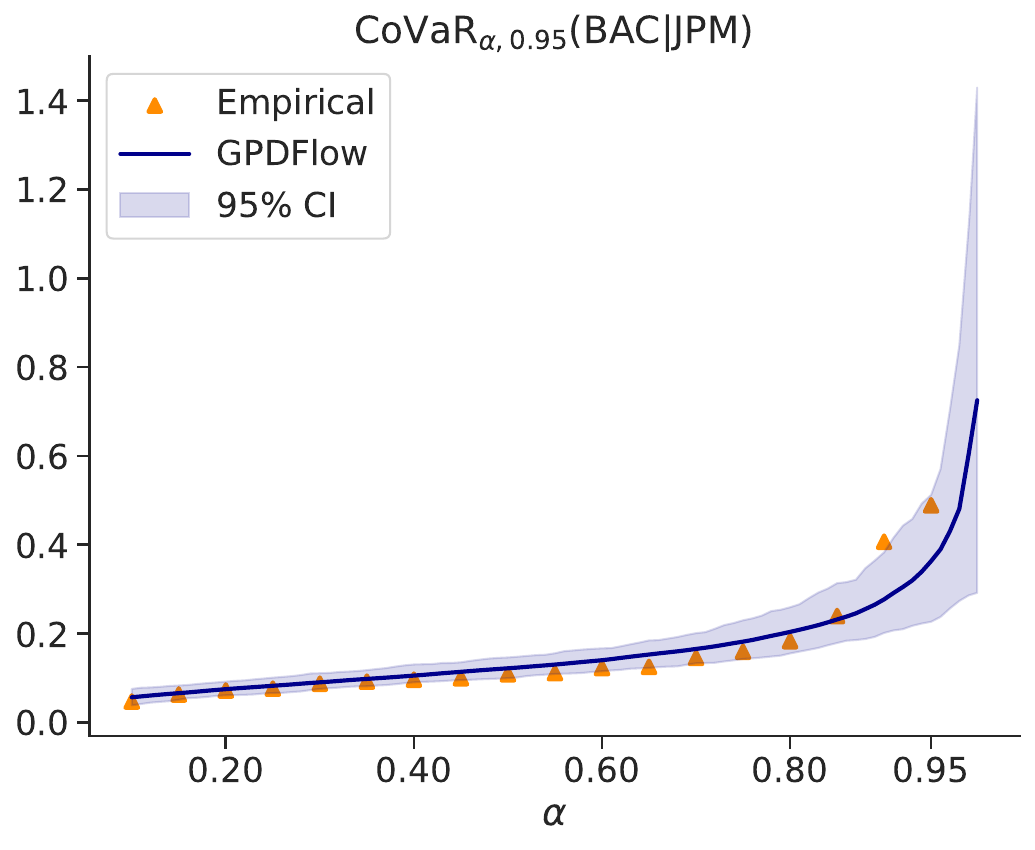}
\includegraphics[width=2.5in]{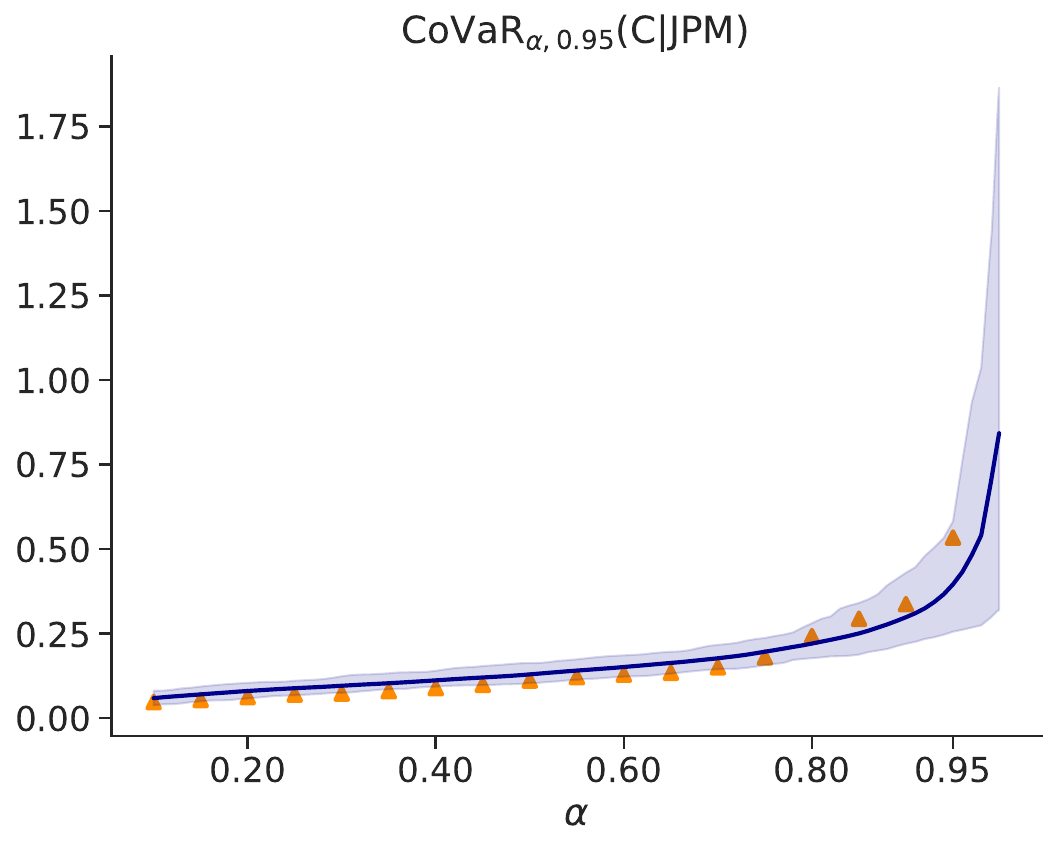}
\includegraphics[width=2.5in]{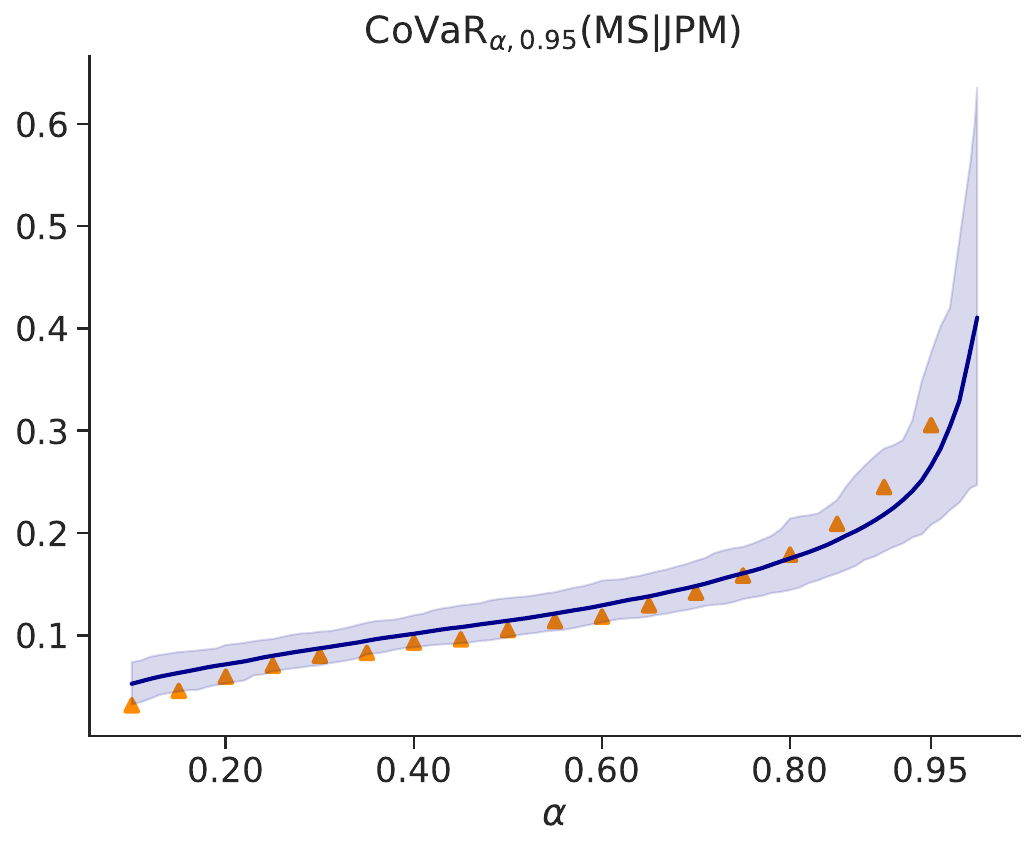}
\includegraphics[width=2.5in]{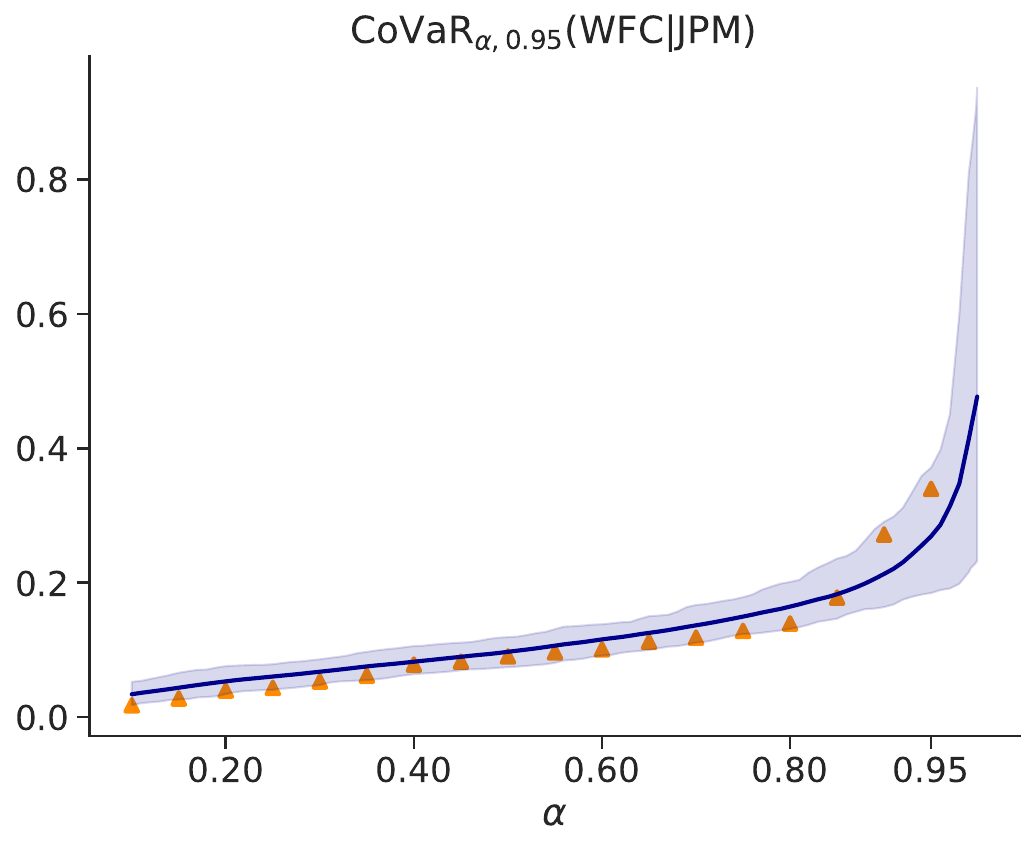}
\end{center}
\caption{\footnotesize{$\text{CoVaR}_{\alpha,\beta}(.|\text{JPM})$ for the four banks, conditioned on JPM being in distress ($\beta=0.95$). The GPDFlow line represents the average CoVaR from 100 Monte Carlo estimations, each derived from the GPDFlow simulated samples that are of the same size as the exceedance data. The 95\% CI is derived from the 2.5\% and 97.5\% quantiles of the Monte Carlo estimations.   }}
\label{fig: CoVaR}
\end{figure}

\section{Discussion}
\label{sec: Discussion}
In this paper, we introduce GPDFlow, a flow-based mGPD model that leverages normalizing flows to capture dependence structures in multivariate threshold exceedances. 
GPDFlow explicitly estimates marginal parameters, allowing direct inference on tail behavior while offering flexible dependence modeling. 
Simulation studies and a real data application show that GPDFlow achieves strong approximation performance for stationary threshold exceedance data.

The GPDFlow framework easily accommodates non-stationary extensions.
Specifically, the marginal parameters $\boldsymbol{\sigma}$ and $\boldsymbol{\gamma}$ can be expressed as functions of covariates through, for instance, MLPs.
Moreover, the normalizing flows component can seamlessly integrate additional covariates by conditioning on them, as demonstrated by architectures like Masked Autoregressive Flows \citep{papamakarios2017masked}.  

Despite these strengths, GPDFlow has three practical limitations that we address here.
Firstly, while GPDFlow benefits from the theoretically robust tail properties of the mGPD framework, it is constrained by the max-stable assumption in \eqref{eq: MDA}.
Under this assumption, GPDFlow always exhibits asymptotic dependence regardless of the structure of the normalizing flows, as evidenced by the positive upper tail coefficient $\chi_{1:d}$ in Propositions \ref{prop: chi_mGPD} and \ref{prop: chi_omega}. 
Although GPDFlow can behave similarly to asymptotic independence (e.g., in the bivariate case, for sufficiently small values of $\mathbb{E}(\exp\{-|T_1 - T_2|\})$ so that the bivariate tail coefficient approaches zero),  the steep decline in $\chi_{1:d}(q)$ as $q \to 1^-$ for asymptotically independent data \citep{huser2019modeling} typically requires setting a very high threshold (e.g., the 0.99-quantile or higher) in GPDFlow for a good approximation.
Consequently, estimation accuracy can decrease due to increased variance and instability caused by limited exceedance data.
This limitation particularly affects high-dimensional scenarios such as spatial extremes, where asymptotic independence often occurs, especially at distant locations.
A practical solution is to apply GPDFlow exclusively to datasets known to exhibit asymptotic dependence (in the spatial context and depending on the data, this could be data from adjacent or close regions) and check if exceedance data are well-approximated by GPDFlow.
A quick way to verify the suitability of GPDFlow is to use Algorithm \eqref{alg: thres_selection} to identify the existence of stable trends in $\chi_{1:d}(q)$ and $\omega_{1:d}(q)$ at high thresholds.

Secondly, estimating uncertainty for marginal parameters in GPDFlow is challenging, despite these parameters being maximum likelihood estimates.
Computing statistical uncertainty typically requires the inverse of the Hessian matrix of all parameters, which becomes computationally impractical given the large parameter space (often thousands) involved in the normalizing flows. 
A bootstrap approach may provide uncertainty estimates but at a substantial computational cost.

Lastly, accurate and stable GPDFlow estimates require threshold exceedance data to be on a similar scale. Divergent data scales can hinder gradient descent convergence or produce inaccurate estimates, a common issue in deep learning models \citep{ioffe2015batch}. To mitigate this, data rescaling can be performed prior to model training. 
Specifically, fitting a univariate GPD to each component's exceedance data and transforming the data using the estimated parameters $\widehat{\sigma}_i$ and $\widehat{\gamma}_i$ via \eqref{eq: transformation} helps improve training stability. 
While this method compromises joint marginal and dependence modeling, it preserves the flexible dependence modeling capability of GPDFlow.

\bigskip


\appendix
\section{ Supplementary material}
\renewcommand{\theequation}{A.\arabic{equation}}
\setcounter{equation}{0}

\renewcommand{\thefigure}{A.\arabic{figure}}
\setcounter{figure}{0}

\subsection{Code and Data}
The code and data required to reproduce the results in Sections~\ref{sec: simulation} and~\ref{sec: application} are freely available at 
\url{https://github.com/hcl516926907/GPDFlow.git}. 
\subsection{Proofs}\label{sec:a_proofs}

\begin{proof}[Proof of Proposition \ref{prop: chi_mGPD}]
The first part of Proposition \ref{prop: chi_mGPD} is a special case $d=2$ of Proposition \ref{prop: chi_omega}, so its proof is omitted here.

For the second part, let $X_1 = \exp\{T_1-\max{(\boldsymbol{T})}\} $ and $X_2 = \exp\{T_2-\max{(\boldsymbol{T}})\}$. 
By the exchangeability assumption of $T_1$ and $T_2$, we have $\mathbb{E}(X_1) = \mathbb{E}(X_2) := c$. 
Since $X_1$ and $X_2$ can only be pair $(1,\; \exp\{-|T_1-T_2|\})$ or $(\exp\{-|T_1-T_2|\}, \;1)$, 
\begin{align*}
    \mathbb{E}(X_1 + X_2) = 1 + \mathbb{E}(\exp\{-|T_1-T_2|\}) = 2c
\end{align*}
Hence $c = \frac{1 + \mathbb{E}(\exp\{-|T_1-T_2|\})}{2} $.
Now 
\begin{align*}
    \chi_{1,2} &=  \mathbb{E}\left(\min \left\{\frac{\exp\{T_1-\max(\boldsymbol{T})\}}{\mathbb{E}(\exp\{T_1-\max(\boldsymbol{T})\})},\frac{\exp\{T_2-\max(\boldsymbol{T})\}}{\mathbb{E}(\exp\{T_2-\max(\boldsymbol{T})\})}\right\}\right)\\
    &= \mathbb{E}\left(\min \left\{\frac{X_1}{\mathbb{E}(X_1)},\frac{X_2}{\mathbb{E}(X_2)}\right\}\right)\\
    &=\frac{1}{c}\mathbb{E}\left(\min \left\{1, \exp\{-|T_1-T_2|\}\right\}\right)\\
    &=\frac{\mathbb{E}(\exp\{-|T_1-T_2|\})}{c}\\
    &= \frac{2\mathbb{E}(\exp\{-|T_1-T_2|\})}{1+\mathbb{E}(\exp\{-|T_1-T_2|\})}
\end{align*}
\end{proof}

\begin{proof}[Proof of Proposition \ref{prop: chi_omega}]
The tail copula  \citep{schmidt2006non} $ 
R(\boldsymbol{x}): [0,\infty)^d \rightarrow [0,\infty)$ defined as 
\begin{align*}
    R(\boldsymbol{x}) = \lim_{n\rightarrow\infty} n\mathbb{P}\{F_1(X_1)>1-x_1/n,\cdots,F_d(X_d)>1-x_d/n\}
\end{align*}
is associated to the stdf $\ell(\boldsymbol{x})$ by inclusion-exclusion formula and can be expressed as 
\begin{align*}
    R(\boldsymbol{x}) = \mathbb{E}(\min(\boldsymbol{x}\boldsymbol{V}))
\end{align*}
by minimum–maximum identity and the same random vector $\boldsymbol{V}$ in \eqref{eq: D_norm} \citep{rootzen2018multivariate}.
By Proposition 7.1 in the same paper, when $q> q^*$, $F_j(X_j)>0, j=1,\cdots,d$ 
\begin{align*}
    \chi_{1:d}(q) &= \frac{\mathbb{P}\{\bigcap_{j=1}^d \{X_j>F_j^{-1}(q)\}\}}{1-q}\\
                  &=\frac{R(q,\cdots,q)}{q} \\
                  &=\mathbb{E}(\min (\boldsymbol{V}))\\
                  &= \mathbb{E}\left(\min \left\{\frac{\exp\{T_1-\max(\boldsymbol{T})\}}{\mathbb{E}(\exp\{T_1-\max(\boldsymbol{T})\})},\cdots,\frac{\exp\{T_d-\max(\boldsymbol{T})\}}{\mathbb{E}(\exp\{T_d-\max(\boldsymbol{T})\})}\right\}\right).
\end{align*}
Similarly,
\begin{align*}
    \omega_{1:d}(q) &= \frac{\mathbb{P}\{\bigcup_{j=1}^d \{X_j>F_j^{-1}(q)\}\}}{1-q}\\
                  &=\frac{\ell(q,\cdots,q)}{q} \\
                  &=\mathbb{E}(\max(\boldsymbol{V}))\\
                  &= \mathbb{E}\left(\max\left\{\frac{\exp\{T_1-\max(\boldsymbol{T})\}}{\mathbb{E}(\exp\{T_1-\max(\boldsymbol{T})\})},\cdots,\frac{\exp\{T_d-\max(\boldsymbol{T})\}}{\mathbb{E}(\exp\{T_d-\max(\boldsymbol{T})\})}\right\}\right)
\end{align*}
\end{proof}

\subsection{Supporting plots}
Below are the plots that are not shown in the main body for space consideration.

\subsubsection{Figure \ref{fig: model_comparison} continued}
\label{sec: appendix_fig5}
\begin{figure}[H]
\begin{center}
\includegraphics[width=2.5in]{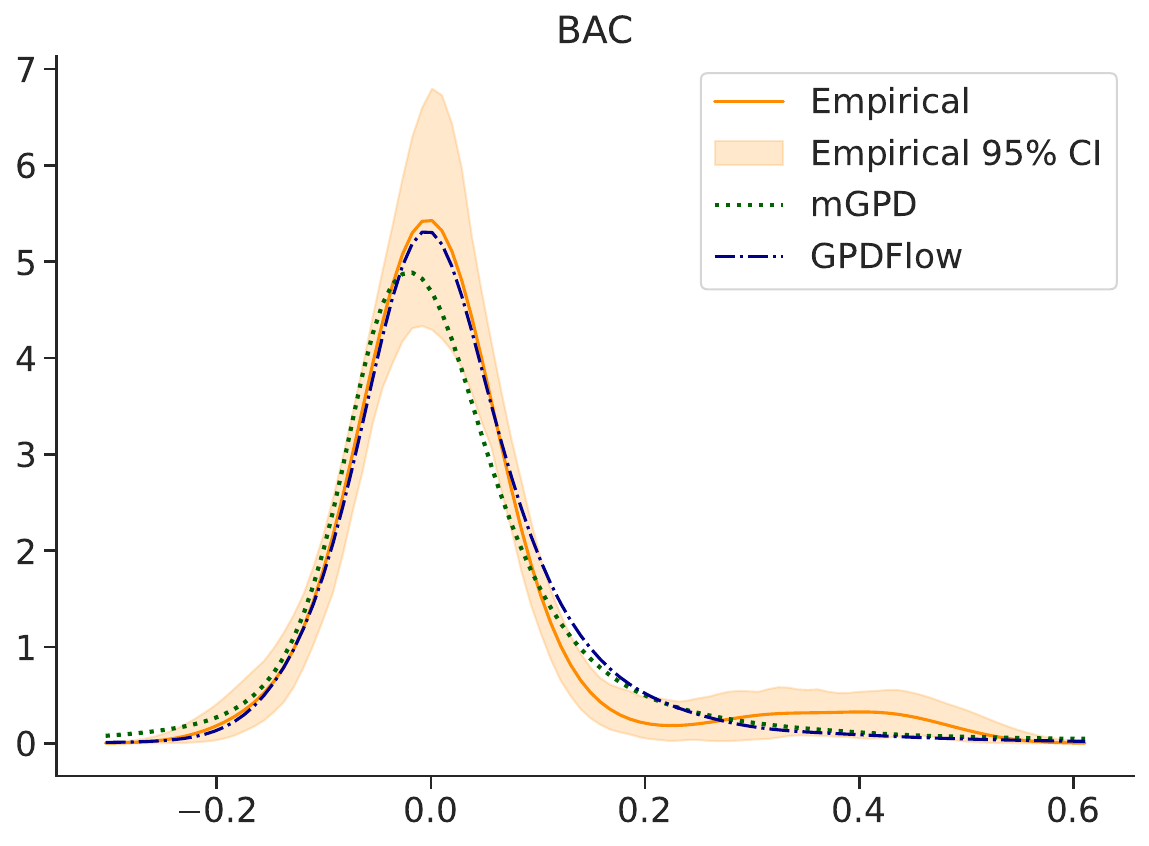}
\includegraphics[width=2.5in]{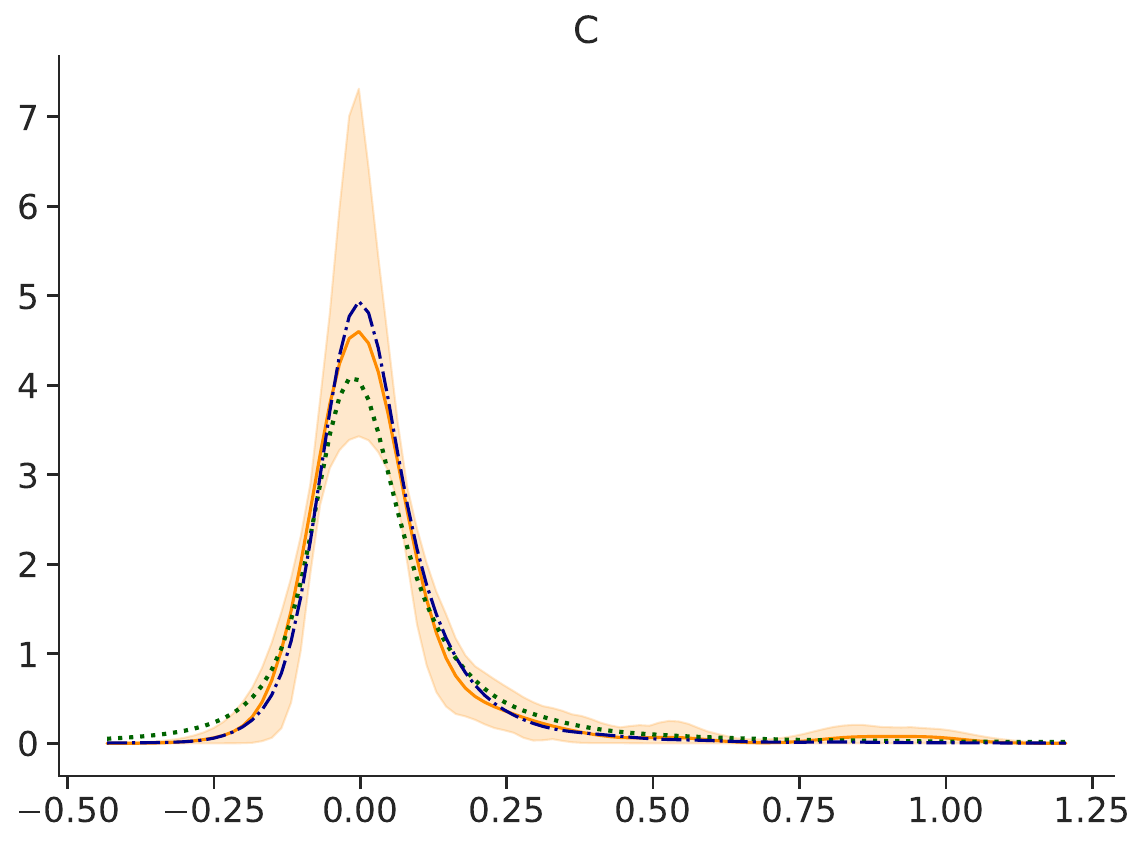}
\includegraphics[width=2.5in]{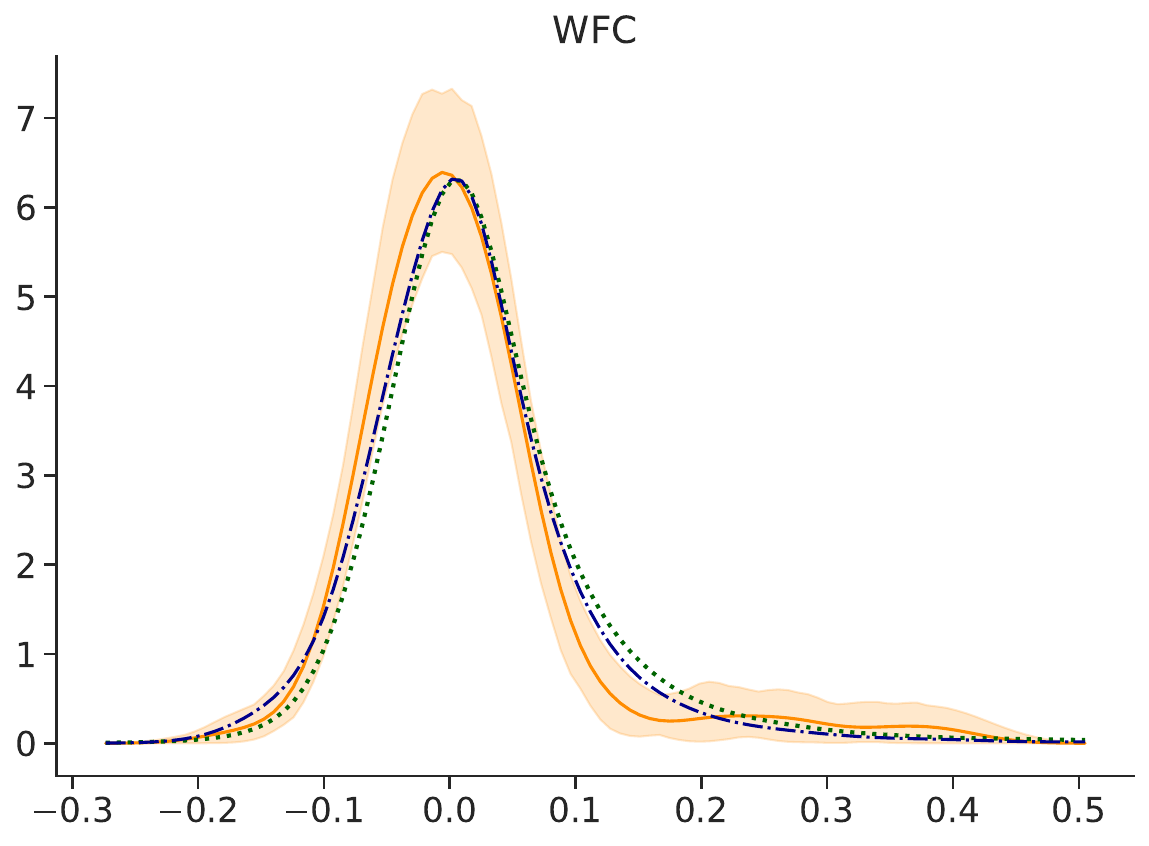}
\end{center}
\caption{\footnotesize{Marginal density estimation comparisons for bank BAC, C and WFC. }}
\end{figure}
\bibliographystyle{apalike}
\bibliography{My_Library}
\end{document}